\documentclass[a4paper,10pt]{article}
\usepackage[toc,page]{appendix}

\newcounter{mnotecount}[section]

\renewcommand{\themnotecount}{\thesection.\arabic{mnotecount}}

\newcommand{\mnote}[1]
{\protect{\stepcounter{mnotecount}}$^{\mbox{\footnotesize $%
\!\!\!\!\!\!\,\bullet$\themnotecount}}$ \marginpar{
\raggedright\tiny\em $\!\!\!\!\!\!\,\bullet$\themnotecount: #1} }

\usepackage{ae}
\usepackage[ansinew]{inputenc}

\usepackage{lscape}
\usepackage{amssymb,amsmath,amsthm}
\usepackage{epsf}
\usepackage{psfrag}
\usepackage{mathrsfs}

\theoremstyle{definition}

\newtheorem{thm}{Theorem}
\newtheorem{lem}{Lemma}

\newtheorem{prop}{Proposition}
\newtheorem{Remark}{Remark}

\newcommand{\metric}{{\bf g}}

\title{The problem of a self-gravitating scalar field with positive cosmological constant}

\author{Jo\~ao L. Costa$^{(1,3)}$, Artur Alho$^{(2)}$ and Jos\'e Nat\'ario$^{(3)}$\\\\
{\small $^{(1)}$Instituto Universitário de Lisboa (ISCTE-IUL), Lisboa, Portugal}\\
{\small $^{(2)}$Centro de Matem\'atica, Universidade do Minho, Gualtar, 4710-057 Braga, Portugal}\\
{\small $^{(3)}$Centro de An\'alise Matem\'atica, Geometria e Sistemas Din\^amicos,}\\
{\small Instituto Superior T\'ecnico, Universidade T\'ecnica de Lisboa, Portugal}
}

\begin{document}

\maketitle

\begin{abstract}
We study the Einstein-scalar field system with positive cosmological constant and spherically symmetric characteristic initial data given on a truncated null cone. We prove well-posedness, global existence  and exponential decay in (Bondi) time, for small data. From this, it follows that initial data close enough to de Sitter data evolves to a causally geodesically complete spacetime (with boundary), which approaches a region of de Sitter asymptotically at an exponential rate; this is a non-linear stability result for de Sitter within the class under consideration, as well as a realization of the cosmic no-hair conjecture.


\end{abstract}

\section{Introduction}

The introduction of a positive cosmological constant $\Lambda$ into the Einstein field equations allows one to model inflation periods (large $\Lambda$) as well as the ``recent" period of accelerated expansion (small $\Lambda$), and consequently plays a central role in modern cosmology. This adds to the relevance of studying initial value problems for the Einstein-matter field equations with positive cosmological constant. For such problems a general framework is provided by the {\em cosmic no-hair conjecture}, which states that generic expanding solutions of Einstein's field equations with a positive cosmological constant approach the de Sitter solution asymptotically. This conjecture as been proved for a variety of matter models and/or symmetry conditions~\cite{Friedrich:1986,Wald:1983,Rendall:2003,Tchapnda:2003,Tchapnda:2005,Ringstrom:2008,Rodnianski:2009,Beyer:2009c,Speck:2011}, but the complexity of the issue makes a general result unattainable in the near future.\footnote{For instance, either by symmetry conditions or smallness assumptions on the initial data the formation of (cosmological) black holes is excluded in all the referred results.}

Here we will consider the spherically symmetric Einstein-scalar field system with positive cosmological constant. This is the simplest, non-pathological matter model with dynamical degrees of freedom in spherical symmetry. By this we mean the following: in spherical symmetry, Birkhoff's theorem completely determines the local structure of electro-vacuum spacetimes, leaving no dynamical degrees of freedom; on the other hand, dust, for instance, is known to develop singularities even in the absence of gravity, i.e.~in a fixed Minkowski background, and consequently is deemed pathological.\footnote{It should be noted that the presence of a positive cosmological constant may counteract the tendency of dust to form singularities.} The self-gravitating scalar field appears then as an appropriate model to study gravitational collapse. This is in fact the original motivation behind the monumental body of work developed by Christodoulou concerning self-gravitating scalar fields with vanishing cosmological constant,\footnote{See the introduction to~\cite{Christodoulou:2008} for a thorough review of Christodoulou's results on spherically symmetric self-gravitating scalar fields.} and it is inspired by these achievements that we proceed to the positive $\Lambda$ case\footnote{Christodoulou's work has also inspired a considerable amount of numerical work, including Choptuik's discovery of critical phenomena \cite{Choptuik:1993} (see also \cite{Gundlach07} and references therein). The case $\Lambda > 0$ seems to be less explored numerically, see however \cite{Brady:1997, Beyer:2009b}.}.

We modify the framework developed in~\cite{Christodoulou:1986} to accommodate the presence of a cosmological constant, thus reducing the full content of the Einstein-scalar field system to a single integro-differential evolution equation. It is then natural, given both the structure of the equation and the domain of the Bondi coordinate system where the reduction is carried out, to consider a characteristic initial value problem by taking initial data on a truncated null cone.

For such an initial value problem we prove well posedness, global existence  and exponential decay in (Bondi) time, for small data. From this, it follows that initial data close enough to de Sitter data evolves, according to the system under consideration, to a causally geodesically complete spacetime (with boundary), which approaches a region of de Sitter asymptotically at an exponential rate; this is a non-linear stability result for de Sitter within the class under consideration and can be seen as a realization of the cosmic no-hair conjecture\footnote{Albeit in a limited sense, since our coordinates do not reach the whole of future infinity (see Figure~\ref{Penrose}). A precise statement of the cosmic no-hair conjecture can be found in \cite{Beyer:2009a}, where it is shown that it follows from the existence of a smooth conformal future boundary}. Also, we  note that the  exponential decay rate obtained, $\lesssim e^{-Hu}$, with  $H=2\sqrt{{\Lambda}/{3}}$, is expected to be sharp\footnote{Although our retarded time coordinate $u$ in \eqref{metricBondi} is different from the standard time coordinate $t$, it coincides with $t$ along the center $r=0$, and hence is close to $t$ in our $r$-bounded domain, thus giving the same exponential decay. For instance, in de Sitter spacetime we have $u=t-\sqrt{3/\Lambda}\,\log\left(1+\sqrt{\Lambda/3}\,r\right)$.} \cite{Rendall:2003}. Moreover, an interesting side effect of the proof of our main results is the generalization, to this non-linear setting, of boundedness of the supremum norm of the scalar field in terms of its initial characteristic data. We refer to Theorem~\ref{mainThm} for a compilation of the main results of this paper.

As was already clear from the study of the uncoupled case~\cite{Costa:2012}, the presence of a positive cosmological constant increases the difficulty of the problem at hand considerably. In fact, a global solution for the zero cosmological constant case was obtained in~\cite{Christodoulou:1986} by constructing a sequence of functions which, for an appropriate choice of Banach space, was a contraction in the full domain; such direct strategy does not work (at least for analogous choices of function spaces) when a positive cosmological constant is considered, since a global contraction is no longer available even in the uncoupled case. Moreover, new difficulties appear in the non-linear problem when passing from zero to a positive cosmological constant: first of all, the incoming light rays (characteristics), whose behavior obviously depends of the unknown, bifurcates into three distinct families, with different, sometimes divergent, asymptotics;\footnote{The use of double null coordinates $(u,v)$, also introduced by Christodoulou for the study of the Einstein-scalar field equations in~\cite{Christodoulou:1991}, would facilitate the handling of the characteristics, which in such coordinates take the form $v=const.$, but in doing so we are no longer able to reduce the full system to a single scalar equation.} this is in contrast with the $\Lambda=0$ case, where all the characteristics approach the center of symmetry at a similar rate. Also, for a vanishing cosmological constant the coefficient of the integral term of the equation decays radially, which is of crucial importance in solving the problem; on the contrary, for $\Lambda>0$ such coefficient grows linearly with the radial coordinate.

To overcome these difficulties we were forced to differ from Christodoulou's original strategy considerably. The cornerstone of our analysis is  a remarkable a priori estimate, the aforementioned result of boundedness in terms of initial data, whose inspiration comes from the uncoupled case~\cite{Costa:2012}. We can then establish a local existence result with estimates for the solution and its radial derivative solely in terms of initial data and constants not depending on the time of existence, which allows us to extend a given local solution indefinitely. The decay results, which in the vanishing cosmological setting are an immediate consequence of the choice of function spaces and the existence of the already mentioned global contraction, here follow by establishing ``energy inequalities'', where the ``energy function'' is given by the supremum norm of the radial derivative of the unknown~\eqref{energyDef}.

To make this strategy work we were forced to restrict our analysis to a finite range of the radial coordinate; one should note nonetheless, that although finite, the results here hold for arbitrarily large radial domains. At a first glance one would expect the need to impose boundary conditions at $r=R$, for $R$ the maximal radius; this turns out to be unnecessary, since for sufficiently large radius the radial coordinate of the characteristics becomes an increasing function of time, and consequently the data at the boundary $r=R$ is completely determined by the initial data (see Figure~\ref{Charact}). This situations parallels that of~\cite{Ringstrom:2008}, where local information in space (here, in a light cone)  allows to obtain global information in time.

A natural consequence of the introduction of a positive cosmological constant is the appearance of a cosmological horizon. In fact, although the small data assumptions do not allow the formation of a black hole event horizon, a cosmological apparent horizon is present from the start, and a cosmological horizon formed; this is of course related to the difficulties mentioned above concerning the dynamics of the characteristics.

\subsection{Previous results}

A discussion of related results in the literature is in order.

The first non-linear stability result for the Einstein equations, without symmetry assumptions, was the non-linear stability of de Sitter spacetime, within the class of solution of the vacuum Einstein equations with positive cosmological constant, obtained in the celebrated work of Friedrich~\cite{Friedrich:1986}.\footnote{This was later generalized to $n+1$ dimensions, $n$ odd, by Anderson~\cite{Anderson:2004}.} This result is based on the conformal method, developed by Friedrich, which avoids the difficulties of establishing global existence of solutions to a system of non-linear hyperbolic differential equations, but seems to be difficult to generalize to Einstein-matter systems. A new, more flexible and PDE oriented approach was recently developed by Ringstr\"om~\cite{Ringstrom:2008} to obtain exponential decay for non-linear perturbations of locally de Sitter cosmological models in the context of the Einstein-nonlinear scalar field system with a positive potential; these far-reaching results have a wide variety of cosmological applications, but break down exactly in the situation covered here, since they assume that the potential $V$ satisfies $V''(0)>0$ (and so cannot be constant).

In the meantime, based on Ringstr\"om's breakthrough, Rodnianski and Speck~\cite{Rodnianski:2009}, and later Speck~\cite{Speck:2011}, proved non-linear stability of FLRW solution with flat toroidal space within the Einstein-Euler system satisfying the equation of state $p=c_s\rho$, $0<c_s<1/3$; exponential decay of solutions close to the flat FLRW was also established therein. In the same context, by generalizing Friedrich's conformal method to pure radiation matter models, L\"ubbe and Kroon~\cite{Lubbe:2011} were able to extend Rodnianski and Speck's non-linear stability result to the pure radiation fluids case, $c_s=1/3$.

\subsection{Main results}

Our main results may be summarized in the following:
\begin{thm}
\label{mainThm}
Let $\Lambda>0$ and $R>\sqrt{3/\Lambda}$. There exists $\epsilon_0>0$, depending on $\Lambda$ and $R$, such that for  $\phi_0\in{\mathcal C}^{k+1}([0,R])$ ($k \geq 1$) satisfying
$$\sup_{0\leq r\leq R }|\phi_0(r)|+\sup_{0\leq r\leq R }|\partial_r\phi_0(r)|<\epsilon_0\;,$$
there exists a unique Bondi-spherically symmetric ${\mathcal C}^k$ solution\footnote{See Section~\ref{Section2} for the precise meaning of a ${\mathcal C}^{1}$ solution of the Einstein-$\Lambda$-scalar field system in Bondi-spherical symmetry.} $(M,\metric,\phi)$ of the Einstein-$\Lambda$-scalar field system~\eqref{fieldEq}, with the scalar field $\phi$ satisfying the characteristic condition
$$\phi_{|_{u=0}}=\phi_0\;.$$
The Bondi coordinates for $M$ have range
$[0,+\infty)\times[0,R]\times \mathbb{S}^2$,
and the metric takes the form~\eqref{metricBondi}. Moreover, we have the following bound in terms of initial data:
\[
\left|\phi\right| \leq \sup_{0\leq r\leq R }\left|\partial_r\left(r\phi_0(r)\right)\right|\;.
\]
Regarding the asymptotics, there exists $\underline{\phi}\in\mathbb{R}$ such that
\[
\label{phiDecay}
\left|\phi(u,r)-\underline{\phi}\right|\lesssim e^{-Hu}\;,
\]
and
\[
\label{metricDecay}
\left|\metric_{\mu\nu}-\mathring{\metric}_{\mu\nu} \right|\lesssim e^{-Hu}\;,
\]
where $H:=2\sqrt{\Lambda/3}$ and $\mathring{\metric}$ is de Sitter's metric in Bondi coordinates, as given in~\eqref{dSBondi}. Finally, the spacetime $(M,\metric)$ is causally geodesically complete towards the future\footnote{A manifold with boundary is geodesically complete towards the future if the only geodesics which cannot be continued for all values of the affine parameter are those with endpoints on the boundary.} and has vanishing final Bondi mass\footnote{See Section~\ref{sec:Bondimass} for the definition of the final Bondi mass in this context.}.
\end{thm}

This result is an immediate consequence of Proposition~\ref{propReduction} and Theorems~\ref{thmGlobal} and~\ref{thmDecay}. Note that, as is the case with the characteristic initial value problem for the wave equation, only $\phi$ needs to be specified on the initial characteristic hypersurface\footnote{Notice however that uniqueness is not expected to hold towards the past.} (as opposed to, say, $\phi$ and $\partial_u\phi$). There is no initial data for the metric functions, whose initial data is fixed by the choice of $\phi_0$. A related issue that may cause confusion is that the vanishing of $\partial_r\phi_0(0)$ is not required to ensure regularity at the center: in fact, the precise condition for $\phi$ to be regular at the center is $\partial_u \phi(u,0) = \partial_r \phi(u,0)$, which is an automatic consequence of the wave equation \eqref{wave}. The reader unfamiliar with these facts should note that, for example, the unique solution of the spherically symmetric wave equation in Minkowski spacetime, $\partial^2_t(r\phi)-\partial^2_r(r\phi)=0$, with initial data $\phi(r,r)=r$, is the smooth function $\phi(t,r)=t$ for $t>r$. 

\section{Einstein-$\Lambda$-scalar field system in Bondi coordinates}\label{Section2}

We will say that a spacetime $(M, \metric)$ is {\em Bondi-spherically symmetric} if it admits a global representation for the metric of the form
\begin{equation}
\label{metricBondi}
 \metric=-g(u,r)\tilde{g}(u,r)du^{2}-2g(u,r)dudr+r^{2}d\Omega^2\;,
\end{equation}
where
$$d\Omega^2=d\theta^2+\sin^2\theta d\varphi^2\;,$$
is the round metric of the two-sphere, and
$$(u,r)\in[0,U)\times[0,R)\;\;,\;\;\; U,R\in\mathbb{R}^+\cup\{+\infty\}\;.$$
If $U$ or $R$ are finite these intervals can also be closed, which corresponds to adding a final light cone $\{u=U\}$ or a cylinder $[0,U) \times \mathbb{S}^2$ as a boundary, in addition to the initial light cone $\{u=0\}$; the metric is assumed to be regular at the center $\{r=0\}$, which is not a boundary. 

The coordinates $(u,r,\theta,\varphi)$ will be called {\em Bondi coordinates}. For instance, the causal future of any point in de Sitter spacetime may be covered by Bondi coordinates with the metric given by
\begin{equation}
 \mathring{\metric}=-\left(1-\frac{\Lambda}{3}r^{2}\right)du^{2}-2dudr+r^{2}d\Omega^{2}\;
\label{dSBondi}
\end{equation}
(see Figure~\ref{Penrose}). Note that this coordinate system does not cover the full de Sitter manifold (which strictly speaking is not Bondi-spherically symmetric), unlike in the asymptotically flat $\Lambda=0$ case. The boundary of the region covered by Bondi coordinates, that is, the surface $u=-\infty$ (which for $\Lambda=0$ would correspond to past null infinity), is an embedded null hypersurface (the cosmological horizon of the observer antipodal to the one at $r=0$). Moreover, the lines of constant $r$, which for $\Lambda=0$ approach timelike geodesics as $r\to+\infty$, here become spacelike for sufficiently large $r$.

\begin{figure}[h!]
\begin{center}
\psfrag{u=+i}{$u=+\infty$}
\psfrag{u=-i}{$u=-\infty$}
\psfrag{u=0}{$u=0$}
\psfrag{H}{$\mathcal{H}$}
\psfrag{H'}{$\mathcal{H'}$}
\psfrag{I+}{$\mathscr{I^+}$}
\epsfxsize=.5\textwidth
\leavevmode
\epsfbox{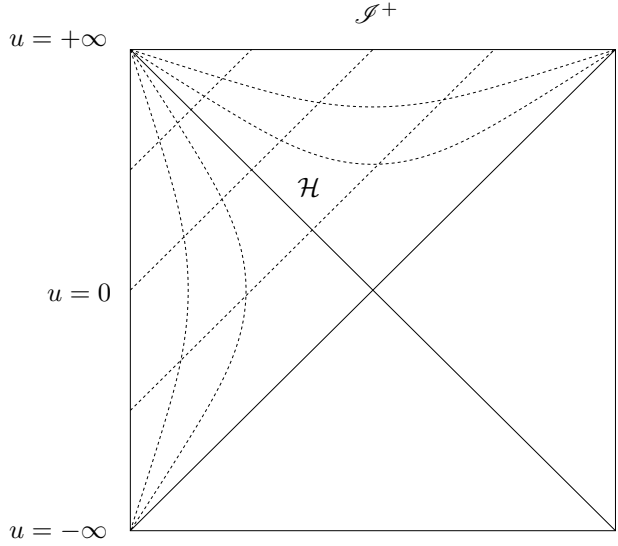}
\end{center}
\caption{Penrose diagram of de Sitter spacetime. The dashed lines $u=\text{constant}$ are the future null cones of points at $r=0$. The cosmological horizon $\mathcal{H}$ corresponds to $r=\sqrt{\frac3\Lambda}$ and future infinity $\mathscr{I^+}$ to $r=+\infty$.} \label{Penrose}
\end{figure}

Although the causal structures of Minkowski and de Sitter spacetimes are quite different, the existence of Bondi coordinates depends solely on certain common symmetries. More precisely, a global representation for the metric of the form~\eqref{metricBondi} can be derived from the following geometrical hypotheses:
\begin{enumerate}
\item[(i)] the spacetime admits a $SO(3)$ action by isometries, whose orbits are either fixed points or $2$-spheres;
\item[(ii)] the orbit space $Q=M/SO(3)$ is a 2-dimensional Lorentzian manifold with boundary, corresponding to the sets of fixed and boundary points in $M$; 
\item[(iii)] the set of fixed points is a timelike curve (necessarily a geodesic), and any point in $M$ is on the future null cone of some fixed point;
\item[(iv)] the {\em radius function}, defined by $r(p):=\sqrt{\text{Area}({\mathcal O}_p)/4\pi}$ (where ${\mathcal O}_p$ is the orbit through $p$), is monotonically increasing along the generators of these future null cones.\footnote{These two last assumptions exclude the Nariai solution, for instance, from our analysis.}
\end{enumerate}

\subsection{The field equations}

The Einstein field equations with a cosmological constant $\Lambda$ are
\begin{equation}
R_{\mu\nu}-\frac{1}{2}R g_{\mu\nu}+\Lambda g_{\mu\nu}=\kappa T_{\mu\nu}\;,
\label{EFEs}
\end{equation}
where $R_{\mu\nu}$ is the Ricci curvature of $\metric$, $R$ the associated scalar curvature and ${T}_{\mu\nu}$ the energy-momentum tensor.
For a (massless) scalar field $\phi$ the energy-momentum tensor is given by
\[
T_{\mu\nu}=\partial_{\mu}\phi\,\partial_{\nu}\phi -\frac{1}{2}\,g_{\mu\nu}\,g^{\alpha\beta}\left(\partial_{\alpha}\phi\right)\left(\partial_{\beta}\phi\right)\;,
\]
and then~\eqref{EFEs} becomes
\begin{equation}
\label{fieldEq}
 R_{\mu\nu}=\kappa\, \partial_{\mu}\phi\,\partial_{\nu}\phi+\Lambda g_{\mu\nu}\;.
\end{equation}
These equations are written for a spacetime metric of the form~\eqref{metricBondi} and a spherically symmetric scalar field in the Appendix. As shown in \cite{Christodoulou:1986}, their full content is encoded in the following three equations: the $rr$ component of the field equations,
\begin{equation}
\qquad\frac{2}{r}\frac{1}{g}\frac{\partial g}{\partial r}=\kappa\left(\partial_{r}\phi\right)^{2}\;;
\label{Einstein_rr}
\end{equation}
the $\theta\theta$ component of the field equations,
\begin{equation}
\frac{\partial}{\partial r}(r\tilde{g})=g\left(1-\Lambda r^{2}\right)\;;
\label{Einstein_theta_theta}
\end{equation}
and the wave equation for the scalar field, 
\begin{equation}
\nabla^{\mu}T_{\mu\nu}=0 \Leftrightarrow \nabla^{\mu}\partial_{\mu}\phi=0\;,
\label{wave}
\end{equation}
which reads
\begin{equation}
\frac{1}{r}\left[\frac{\partial}{\partial u}-\frac{\tilde{g}}{2}\frac{\partial}{\partial r}\right]\frac{\partial}{\partial r}\left(r\phi\right)=\frac{1}{2}\left(\frac{\partial\tilde{g}}{\partial r}\right)\left(\frac{\partial\phi}{\partial r}\right)\;.
\label{WaveEquation}
\end{equation}

\subsection{Christodoulou's framework for spherical waves}

Integrating~\eqref{Einstein_rr} with initial condition
\[
g(u,r=0)=1
\]
(so that we label the future null cones by the proper time of the free-falling observer at the center\footnote{This differs from Christodoulou's original choice, which was to use the proper time of observers at infinity.}) yields
\begin{equation}
 g=e^{\frac{\kappa}{2}\int^{r}_{0}s(\partial_{s}\phi)^{2}ds}.
\label{gtermsPhi}
\end{equation}

Given any continuous function $f=f(u,r)$ we define its {\em average function} by
\begin{equation}
\label{meanDef}
\bar{f}(u,r):=\frac{1}{r}\int^{r}_{0}f\left(u,s\right)ds\;,
\end{equation}
for which the following identity holds:
\begin{equation}
\label{meanDif}
\frac{\partial\bar{f}}{\partial r}=\frac{f-\bar{f}}{r}\;.
\end{equation}
Using the regularity condition
\[
\lim_{r\rightarrow0}r\tilde{g}=0\;,
\]
implicit in our definition of Bondi-spherically symmetric spacetime, we obtain by integrating~\eqref{Einstein_theta_theta}:
\begin{equation}
 \tilde{g}=\frac{1}{r}\int^{r}_{0}g\left(1-\Lambda s^{2}\right)ds=\overline{g\left(1-\Lambda r^{2}\right)}=\bar{g}-\frac{\Lambda}{r}\int^{r}_{0}gs^{2}ds\;.
\label{barg_terms_g}
\end{equation}
Following~\cite{Christodoulou:1986} we introduce
\[
h:={\partial_r}\left(r\phi\right)\;.
\]
Assuming $\phi$ continuous, which implies
\[
\lim_{r\rightarrow0}r\phi=0\;,
\]
we have
\begin{equation}
 \phi=\frac{1}{r}\int^{r}_{0}h\left(u,s\right)ds=\bar{h}\quad\text{and}\quad\frac{\partial\phi}{\partial r}=\frac{\partial\bar{h}}{\partial r}=\frac{h-\bar{h}}{r}\;,
\label{Definitionh}
\end{equation}
and so \eqref{gtermsPhi} reads
\begin{equation}
\label{g}
g(u,r)=\exp\left({\frac{\kappa}{2}\int^{r}_{0}\frac{\left(h-\bar{h}\right)^{2}}{s}ds}\right)\;.
\end{equation}

Now, defining the differential operator
\[
D:=\frac{\partial}{\partial u}-\frac{\tilde{g}}{2}\frac{\partial}{\partial r}\;,
\]
whose integral lines are the incoming light rays (with respect to the observer at the center $r=0$), and using \eqref{Definitionh} together with \eqref{Einstein_theta_theta}, the wave-equation \eqref{WaveEquation} is rewritten as the integro-differential equation
\begin{equation}
\label{mainEqOriginal} 
 Dh=G\left(h-\bar{h}\right)\;,
\end{equation}
where we have set
\begin{eqnarray}
\label{defG}
G&:=&\frac{1}{2}\partial_r\tilde{g} \\
\label{defG2}
&=&\frac{1}{2r}\left[(1-\Lambda r^2)g-\overline{(1-\Lambda r^2)g}\right] \\
\label{defG3}
&=&\frac{\left(g-\bar{g}\right)}{2r}+\frac{\Lambda}{2r^{2}}\int^{r}_{0}gs^{2}ds-\frac{\Lambda}{2}rg\;.
\end{eqnarray}

Thus we have derived the following:
\begin{prop}
\label{propReduction}
For Bondi-spherically symmetric spacetimes~\eqref{metricBondi}, the Einstein-scalar field system with cosmological constant~\eqref{fieldEq} is equivalent to the integro-differential equation~\eqref{mainEqOriginal}, together with \eqref{barg_terms_g}, \eqref{Definitionh}, \eqref{g} and \eqref{defG}.
\end{prop}
We will also need an evolution equation for $\partial_rh$ given a sufficiently regular solution of~\eqref{mainEqOriginal}: using
\[
\left[D,\partial_{r}\right]=G\partial_{r}\;,
\]
differentiating~\eqref{mainEqOriginal}, and assuming that we are allowed to commute partial derivatives, we obtain
\begin{equation}
\begin{aligned}
 D\partial_{r}h-2G\partial_{r}h&=-J\,\partial_{r}\bar{h}\;,
\end{aligned}
\label{D_partial_h}
\end{equation}
where
\begin{eqnarray}
\label{defJ}
J:&=&G-r\partial_rG
\\
\label{defJ2}
&=&3G+\Lambda g r+(\Lambda r^{2}-1)\frac{1}{2}\frac{\partial g}{\partial r}\;.
\end{eqnarray}

\section{The mass equation}\label{sec:Bondimass}

Consider a Bondi-spherically symmetric ${\mathcal C}^k$ solution of \eqref{EFEs} on a domain $(u,r) \in [0,U) \times [0,R]$ (with $R > \sqrt{3/\Lambda}$). From equations \eqref{Einstein_theta_theta} and \eqref{g} it is clear that $r\tilde{g}$ is increasing in $r$ for $r < \sqrt{1/\Lambda}$ and decreasing for $r > \sqrt{1/\Lambda}$. On the other hand, equation \eqref{barg_terms_g} implies that $\tilde{g}(u,r)$ approaches $-\infty$ as $r \to +\infty$. Therefore there exists a unique $r=r_c(u)> \sqrt{1/\Lambda}$ where $\tilde{g}(u,r)$ vanishes. This defines precisely the set of points where $\frac{\partial}{\partial u}$ is null, and hence the curve $r=r_c(u)$ determines an apparent (cosmological) horizon. Since $g$ is increasing in $r$, we have from \eqref{barg_terms_g}
\[
\tilde{g}(u,r) \leq g(u,\sqrt{1/\Lambda}) \frac1r \int_0^r \left(1 - \Lambda s^2\right) ds = g(u,\sqrt{1/\Lambda}) \left(1 - \frac{\Lambda r^2}3 \right).
\]
Therefore the radius of the apparent cosmological horizon is bounded by
\[
\sqrt{\frac{1}{\Lambda}}<r_{c}(u)\leq\sqrt{\frac{3}{\Lambda}}
\]
for all $u$. From \eqref{Einstein_theta_theta} it is then clear that $\frac{\partial \tilde{g}}{\partial r} < 0$ for $r=r_c(u)$, and so by the implicit function theorem the function $r_c(u)$ is ${\mathcal C}^k$. From the $uu$ component of \eqref{fieldEq} (equation \eqref{Einstein_uu} in the Appendix), we obtain
\[
\frac{g}{r}\frac{\partial}{\partial u}\left(\frac{\tilde{g}}{g}\right)=\kappa (\partial_{u}\phi)^{2}
\]
when $\tilde{g}=0$, showing that $\frac{\tilde{g}}{g}$ is nondecreasing in $u$, and so $r_c(u)$ must also be nondecreasing. In particular the limit
\[
r_1:= \lim_{u\to U}r_{c}(u)
\]
 exists, and $\sqrt{1/\Lambda}<r_{1}\leq\sqrt{3/\Lambda}$.

We introduce the renormalized Hawking mass function\footnote{This function is also known as the ``generalized Misner-Sharp mass".} \cite{Nakao:1995, Maeda:2007}
\begin{equation}
m(u,r)=\frac{r}{2}\left(1-\frac{\tilde{g}}{g}-\frac{\Lambda}{3}r^{2}\right),
\label{massfunction}
\end{equation}
which measures the mass contained within the sphere of radius $r$ at retarded time $u$, renormalized so as to remove the contribution of the cosmological constant and make it coincide with the mass parameter in the case of the Schwarzschild-de Sitter spacetime. This function is zero at $r=0$, and from \eqref{Einstein_rr}, \eqref{Einstein_theta_theta} we obtain
\[
\frac{\partial m}{\partial r} = \frac{\kappa r^2 \tilde{g}}{4g} (\partial_{u}\phi)^{2},
\]
implying that $m(u,r)\geq 0$ for $r \leq r_{c}(u)$. We have
\[
m(u,r_{c}(u)) = \frac{r_c(u)}{2}\left(1-\frac{\Lambda}{3}{r_c(u)}^{2}\right),
\]
whence
\[
\frac{d}{du} m(u,r_{c}(u)) = \frac{\dot{r}_c(u)}{2}\left(1-\Lambda{r_c(u)}^{2}\right) \leq 0,
\]
and so $m(u,r_{c}(u))$ is a nonincreasing function of $u$. Therefore the limit
\[
M_1:=\lim_{u\to U}m(u,r_{c}(u))=\frac{r_{1}}{2}\left(1-\frac{\Lambda}{3}r^{2}_{1}\right)
\]
exists, and from $\sqrt{1/\Lambda}<r_{1}\leq\sqrt{3/\Lambda}$ we have $0 \leq M_{1}<1/\sqrt{9\Lambda}$. We call this limit the {\em final Bondi mass}. Note that, unlike the usual definition in the asymptotically flat case, where the limit is taken at $r=+\infty$, here we take the limit along the apparent cosmological horizon; the reason for doing this is that $r \leq R$ in our case.


\section{Basic Estimates}

\newcommand{\our}{\mathcal{C}^{0}_{U,R}}
\newcommand{\Xour}{X_{U,R}}

Given $U,R>0$, let $\mathcal{C}^{0}_{U,R}$ denote the Banach space
$\left(\mathcal{C}^{0}([0,U]\times[0,R]),\|\cdot\|_{\mathcal{C}^{0}_{U,R}}\right)$, where
\begin{equation*}
 \left\|f\right\|_{\mathcal{C}^{0}_{U,R}}:=\sup_{(u,r)\in[0,U]\times[0,R]}\left|f(u,r)\right|,
\end{equation*}
and let $X_{U,R}$ denote the Banach space of functions
which are continuous and have continuous partial derivative with respect to $r$, normed by
\begin{equation*}
 \|f\|_{X_{U,R}}:=\|f\|_{\our}+\|\partial_r f\|_{\our}\;.
\end{equation*}
For functions defined on $[0,R]$ we will denote $C^0([0,R])$ by $C^0_R$, $C^1([0,R])$ by $X_R$, and will also use these notations for the corresponding norms.

\vspace{0.2cm}

For $h\in\our$ we have
\[
  \left|\bar{h}(u,r)\right|
  \leq\frac{1}{r}\int^{r}_{0}\left|h(u,s)\right|ds\leq\frac{1}{r}\int^{r}_{0}\|h\|_{\mathcal{C}^{0}_{U,R}}\,ds
  =\|h\|_{\mathcal{C}^{0}_{U,R}}
\]
and if $h \in\Xour$ we can estimate
\begin{equation}
\begin{aligned}
\left|(h-\bar{h})(u,r)\right|&=\left|\frac{1}{r}\int^{r}_{0}\left(h(u,r)-h(u,s)\right)ds\right|=\left|\frac{1}{r}\int^{r}_{0}\int^{r}_{s}\frac{\partial h}{\partial\rho}(u,\rho)d\rho\, ds\right| \\
&\leq\frac{1}{r}\int^{r}_{0}\int^{r}_{s}\|\partial_{r}h\|_{\mathcal{C}^{0}_{U,R}}\,d\rho\,ds=\frac{r}{2}\,\|\partial_{r}h\|_{\mathcal{C}^{0}_{U,R}}\;.
\end{aligned}
\label{Est_h_bar_h}
\end{equation}
Thus
\begin{equation*}
 \frac{\kappa}{2}\int_0^{R}\frac{\left(h-\bar{h}\right)^{2}}{r}dr\leq\frac{\kappa}{16} \|\partial_{r}h\|_{\our}^{2}R^{2}\,,
\end{equation*}
and by~\eqref{g} we get
\begin{equation}
\label{defK}
g(u,0)=1\leq g(u,r)\leq K:=\exp\left(\frac{\kappa}{16}\|\partial_{r}h\|_{\our}^{2}R^{2}\right).
\end{equation}

\subsection{The characteristics of the problem}

The integral curves of $D$, which are the incoming light rays, are the characteristics of the problem.
These satisfy the ordinary differential equation,
\begin{equation}
 \frac{dr}{du}=-\frac{1}{2}\tilde{g}(u,r).
\label{Characteristic_ODE}
\end{equation}
To simplify the notation we shall denote simply by $u\mapsto r(u)$, the solution to~\eqref{Characteristic_ODE},
satisfying $r(u_1)=r_1$.
However it should be always kept in mind that $r(u)=r(u;u_1,r_1)$. \\

\newcommand{\myK}{K}

Using~\eqref{defK} we can estimate $\tilde{g}$, given by~\eqref{barg_terms_g},  and consequently the solutions to the characteristic
equation~\eqref{Characteristic_ODE}: In fact, for $r\leq\frac{1}{\sqrt{\Lambda}}\Rightarrow1-\Lambda r^{2}\geq0$ we get
\begin{equation*}
\begin{aligned}
 \tilde{g}\geq\frac{1}{r}\int^{r}_{0}(1-\Lambda s^{2})ds=1-\frac{\Lambda}{3}r^{2}\geq  1-{\myK}\frac{\Lambda}{3}r^{2}\;.
\end{aligned}
\end{equation*}
For  $r\geq\frac{1}{\sqrt{\Lambda}}$ we have
\begin{equation*}
\begin{aligned}
\tilde{g}(u,r)
&\geq
\frac{1}{r}\int^{\frac{1}{\sqrt{\Lambda}}}_{0}\left(1-\Lambda s^{2}\right)ds+\frac{{\myK}}{r}\int^{r}_{\frac{1}{\sqrt{\Lambda}}}\left(1-\Lambda s^{2}\right)ds
 \\
&= \frac{2}{3\sqrt{\Lambda}r}\left(1-{\myK}\right)+{\myK}\left(1-\frac{\Lambda}{3}r^{2}\right)
\\
&\geq \frac{2}{3}\left(1-{\myK}\right)+{\myK}\left(1-\frac{\Lambda}{3}r^{2}\right)\;.
\end{aligned}
\end{equation*}
We then see that the following estimate holds for all $r\geq 0$:
\begin{equation}
\tilde{g}\geq 1-\frac{{\myK}\Lambda}{3}r^{2}\;.
\label{Est_tilde_g}
\end{equation}
The same kind of reasoning also provides the upper bound
\begin{equation}
\label{tildegUpper}
\tilde{g}\leq \myK-\frac{\Lambda}{3}r^{2}\;.
\end{equation}

From~\eqref{Characteristic_ODE} and~\eqref{Est_tilde_g} we now obtain the following differential inequality
\begin{equation}
\frac{dr}{du}\leq-\frac{1}{2}+\frac{\Lambda {\myK}}{6}r^2\;.
\label{Estimates_Characteristic_ODE}
\end{equation}
Denoting
$$
\alpha=\frac{1}{2}\sqrt{\frac{\Lambda {\myK}}{3}} \quad\text{ and }\quad r^{-}_{c}=\sqrt{\frac{3}{\Lambda {\myK}}} \;,
$$
where $r^{-}_{c}$ is the positive root of the polynomial in~\eqref{Est_tilde_g}, the solution $r^-(u)$ of the differential equation obtained from~\eqref{Estimates_Characteristic_ODE} (by replacing the inequality with an equality) satisfying $r^-(u_1)=r_1< r^{-}_{c}$ is given by
\[
r^-(u)=\frac{1}{2\alpha}\tanh{\left\{\alpha(c^--u)\right\}}\;,
\]
for some $c^{-}=c^{-}(u_1,r_1)$; by a basic comparison principle it then follows that whenever $r(u_1)=r_1< r^{-}_{c}$ we have
\begin{equation}
r(u)\geq\frac{1}{2\alpha}\tanh{\left\{\alpha(c^--u)\right\}}\;\;\;,\;\; \forall u\leq u_1\;.
\label{Est_char_loc}
\end{equation}

Denote the positive root of the polynomial in~\eqref{tildegUpper} by
$$r^{+}_{c}=\sqrt{\frac{3{\myK}}{\Lambda }}\;;$$
then, for appropriate choices (differing in each case) of $c^{-}=c^{-}(u_1,r_1)$ and $c^{+}=c^{+}(u_1,r_1)$, similar reasonings based on comparison principles give the following global estimates for the characteristics (see also Figure~\ref{Charact}):
\begin{itemize}
\item {\bf Local region} ($r_1<r^{-}_{c}$):
\end{itemize}
\begin{equation}
 \frac{1}{2\alpha}\tanh{\left\{\alpha(c^{-}-u)\right\}}\leq r(u)\leq\frac{{\myK}}{2\alpha}\tanh{\left\{\alpha(c^{+}-u)\right\}}\;\;\;,\;\; \forall u\leq u_1\;.
\label{charLoc}
\end{equation}
\begin{itemize}
 \item {\bf Intermediate region} ($r^{-}_{c}\leq r_1< r^{+}_{c}$):
\end{itemize}
\begin{equation}
\frac{1}{2\alpha}\coth{\left\{\alpha(c^{-}-u)\right\}}\leq r(u)\leq\frac{{\myK}}{2\alpha}\tanh{\left\{\alpha(c^{+}-u)\right\}}\;\;\;,\;\; \forall
u\leq u_1\;.
\label{charInt}
\end{equation}
\begin{itemize}
 \item {\bf Cosmological region} ($r\geq r^{+}_{c}$):
\end{itemize}
\begin{equation}
 \frac{1}{2\alpha}\coth{\left\{\alpha(c^{-}-u)\right\}}\leq r(u)\leq\frac{{\myK}}{2\alpha}\coth{\left\{\alpha(c^{+}-u)\right\}}\;\;\;,\;\; \forall u\leq u_1\;.
 \label{charCosm}
\end{equation}

In particular, for $r(u_1)=r_1\geq r^{-}_{c}$ we obtain
\begin{equation}
 r(u)\geq r^{-}_{c}>0\;\;\;,\;\; \forall u\leq u_1.
\label{Est_char_cosm}
\end{equation}

\begin{figure}[h!]
\begin{center}
\psfrag{u}{$u$}
\psfrag{r}{$r$}
\psfrag{rc+}{$r_c^+$}
\psfrag{rc-}{$r_c^-$}
\psfrag{(u,r)}{$(u_1,r_1)$}
\epsfxsize=.6\textwidth
\leavevmode
\epsfbox{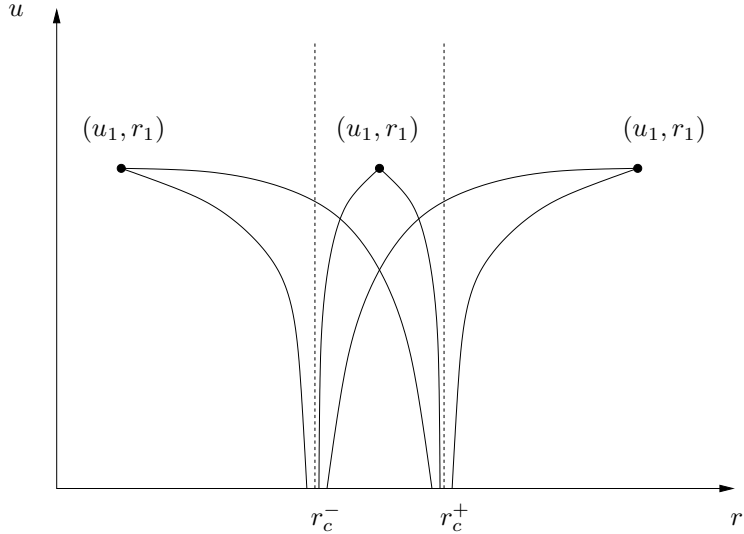}
\end{center}
\caption{Bounds for the characteristics through the point $(u_1,r_1)$ in the local ($r_1 < r_c^-$), intermediate ($r_c^- \leq r_1 < r_c^+$) and cosmological ($r_1 \geq r_c^+$) regions.} \label{Charact}
\end{figure}

\subsection{Lemma 1}

The purpose of this section is to prove the following lemma:

\begin{lem}
\label{Lemma1}
Let $\Lambda>0$ and $R>0$.
There exists $x^{*}=x^*(\Lambda,R)>0$ and constants $C_i=C_i(x^*,\Lambda,R)>0$, such that if $\left\|h\right\|_{X_{U,R}}\leq x^{*}$, then\footnote{As usual, $O(x^*)$ means a bounded function of $x^*$ times $x^*$ in some neighborhood of $x^*=0$.}
%
\begin{equation}
\label{Gbound1}
G<-C_1r\;,\;\; C_1=\frac{\Lambda}{3}+O(x^*)\;,
\end{equation}
\begin{equation}
\label{Gbound2}
|G|<C_2r\;, \\
\end{equation}
\begin{equation}
\label{boundJ}
|J|<C_3r\;, \;\; C_3=O(x^*),
\end{equation}
and, for any $u_1\geq0$ and $r_1\leq R$,
\begin{equation}
\label{boundInt}
\int^{u_1}_{0}\exp\left({\int^{u_1}_{u}2 G(v,r(v))dv}\right)du \leq C_4\,,
\end{equation}
where $r(u)=r(u;u_1,r_1)$ is the characteristic through $(u_1,r_1)$.
\end{lem}

\begin{Remark}
We stress the fact that while allowed to depend on $R$ the constants do not depend on any parameter associated with the $u$-coordinate.
\end{Remark}

\begin{proof}

We have, from~\eqref{defK},
\begin{equation}
\label{g*}
1\leq g \leq K^*:=\exp\left(\frac{\kappa}{16}(x^*)^2R^{2}\right) \;.
\end{equation}

Differentiating~\eqref{g} while using~\eqref{Est_h_bar_h} and \eqref{g*} leads to\footnote{From now on we will use the notation $f\lesssim g$ meaning that  $f\leq Cg$, for $C\geq0$ only allowed to depend on the fixed parameters $\Lambda$ and $R$.}
\begin{equation}
0\leq\frac{\partial g}{\partial r}\lesssim g\frac{\left(h-\bar{h}\right)^{2}}{r}\lesssim K^*(x^*)^2r,
\label{Est_partial_g}
\end{equation}
and consequently
\begin{equation*}
\begin{aligned}
0\leq\left(g-\bar{g}\right)(u,r)&= \frac{1}{r}\int^{r}_{0}\left(g(u,r)-g(u,s)\right)ds
\\
&=\frac{1}{r}\int^{r}_{0}\left\{\int^{r}_{s}\frac{\partial g}{\partial \rho}(u,\rho)d\rho\right\}ds
\\
&\lesssim \frac{1}{r}\int_0^{r}\int_s^{r}K^*(x^*)^2\rho \, d\rho \, ds
\\
&\lesssim K^*(x^*)^2r^{2}.
\end{aligned}
\end{equation*}

From this estimate,~\eqref{defG3} and~\eqref{g*} we see that
\begin{equation}
\label{Gbound0}
\left(\frac{\Lambda}{6}-\frac{\Lambda}{2}K^*\right)r \leq G\leq \left[K^*\left(C(x^*)^2+\frac{\Lambda}{6}\right)-\frac{\Lambda}{2}\right]r
\end{equation}
for some constant $C>0$ depending only on $\Lambda$ and $R$. Since $K^*\rightarrow 1$ as $x^*\rightarrow 0$,~\eqref{Gbound1} then follows by choosing $x^*$ appropriately small.
Also, inequality~\eqref{Gbound2} is immediate.

From~\eqref{defJ2},~\eqref{Est_partial_g} and~\eqref{Gbound0} we now obtain~\eqref{boundJ}.\\

 To prove~\eqref{boundInt} we start by using~\eqref{Gbound1} to obtain
\begin{equation*}
\int_0^{u_1}e^{\int_u^{u_1}2G(v,r(v))dv}du\leq\int_0^{u_1}e^{-2C_1\int^{u_1}_{u}r(v)dv}du\;.
\end{equation*}
If $r_1< r^{-}_{c}=\sqrt{\frac{3}{\Lambda {\myK}}}$ then~\eqref{Est_char_loc} holds and we then have
\begin{equation*}
\begin{aligned}
 -2C_1\int^{u_{1}}_{u}r(v)dv&\leq-\frac{C_1}{\alpha}\int^{u_{1}}_{u}\tanh{(\alpha({c^-}-v))}dv \\
                             &=\frac{C_1}{\alpha^2}\ln{\left(\frac{\cosh{(\alpha({c^-}-u_1))}}{\cosh{(\alpha ({c^-}-u))}}\right)}.
\end{aligned}
\end{equation*}
Since
\begin{equation*}
\frac{\cosh{(\alpha({c^-}-u_1))}}{\cosh{(\alpha({c^-}-u))}}\leq 2e^{\alpha (u-u_1)}
\end{equation*}
and
$$\frac{1}{2}\sqrt{\frac{\Lambda }{3}}\leq\alpha= \frac{1}{2}\sqrt{\frac{\Lambda {\myK}}{3}}\lesssim \sqrt{K^*} \;,$$
we obtain
\begin{equation}
\begin{aligned}
\label{intLocalRegion}
\int^{u_1}_{0}e^{-2{C_1}\int^{u_1}_{u}r(v)dv}du&\leq2^{{C_1}/{\alpha^2}}\int^{u_1}_{0}e^{\frac{C_1}{\alpha}(u-u_1)}du \\
                                                              &\leq 2^{{C_1}/{\alpha^2}}\frac{\alpha}{C_1}\left[1-e^{-\frac{C_1}{\alpha}u_1}\right]\leq 2^{{C_1}/{\alpha^2}}\frac{\alpha}{C_1}\leq C_4(x^*,\Lambda,R)\;,
\end{aligned}
\end{equation}
as desired.

If $r_1\geq r_c^-$, we have~\eqref{Est_char_cosm} which gives
$$\int_0^{u_1}e^{-2C_1\int^{u_1}_{u}r(v)dv}du\leq \int_0^{u_1}e^{-2C_1r_c^-(u_1-u)}du\leq \frac{1}{2C_1r_c^-}\left[1-e^{-2C_1r_c^-u_1}\right]\leq\frac{2\alpha}{C_1}\leq C_4(x^*,\Lambda,R)\;,$$
which completes the proof of the lemma.

\end{proof}


\section{Controlled local existence}

Local existence will be proven by constructing a contracting sequence of solutions to related linear problems. Given a sequence $\{h_n\}$ we will write $g_n:=g(h_n)$, $G_n:=G(h_n)$,  etc, for the quantities~\eqref{g},~\eqref{defG}, etc, obtained from $h_n$; for a given $h_n$ the corresponding differential operator will be denoted by
$$D_n=\partial_u-\frac{\tilde{g}_n}{2}\partial_r\,,$$
and the associated characteristic through $(u_1,r_1)$ by $\chi_n=\chi_n(u)=(u,r_n(u;u_1,r_1))$; as before, we will drop the explicit dependence on initial conditions when confusion is unlikely to arise.

With these notational issues settled we are ready to prove the following fundamental result:
\begin{lem}
\label{Lemma2}
Let $\Lambda>0$, $R>\sqrt{\frac{3}{\Lambda}}$ and $h_0\in \mathcal{C}^1([0,R])$.
There exists $x^*=x^*(\Lambda,R)>0$ and $C^*=C^*(x^*,\Lambda,R)>0$ such that if
\[
\|h_{0}\|_{X_{R}}\leq\frac{x^{*}}{1+C^{*}}\;,
\]
then the sequence $\{h_{n}\}_{n\in\mathbb{N}_0}$  defined by $h_0(u,r)=h_0(r)$ and
\begin{equation*}
 \left\{
\begin{array}{l}
 D_{n}h_{n+1}-G_{n}h_{n+1}=-G_{n}\bar{h}_{n}\\
 h_{n+1}(0,r)=h_{0}(r)\;,
\end{array}
\right.
\end{equation*}
is in $C^1([0,R]\times[0,U])$ and satisfies
\begin{eqnarray}
\label{G<0}
G_{n}&\leq&0\;,
\\
\label{h<h0}
 \|h_{n}\|_{C^{0}_{U,R}}&=&\|h_{0}\|_{C^{0}_{R}}\;,
\\
\label{X<X0}
 \|h_{n}\|_{X_{U,R}}&\leq&(1+C^{*})\|h_{0}\|_{X_{R}}\;,
\end{eqnarray}
for all $n\in\mathbb{N}_0$ and all $U\geq 0\;.$
\end{lem}

\begin{Remark}
We stress the fact that $C^*$ does not depend on either $U$ or $n$.
\end{Remark}

\begin{proof}

The proof is by induction. That the conclusions follow for the $0^{th}$ term is immediate, with~\eqref{G<0} obtained from Lemma~\ref{Lemma1} by setting $x^*$ accordingly small. Assume that $h_n$ satisfies all the conclusions of the lemma.
In particular, since we have $h_n\in C^1([0,U]\times[0,R])$ we see, from the respective definitions, that $\bar{h}_n$, $g_n$ and $\tilde g_n$ are $C^1$ for $r\neq 0$; regularity at the origin then follows by inserting the first order Taylor expansion in $r$ of $h_n$, centered at $r=0$, in the definitions of $\bar{h}_n$, then $g_n$ and finally $\tilde g_n$  . Later in the proof we will also need $\partial_rG_n$ to be well defined and continuous in the domain under consideration; this follows by using the previous referred expansions in equations~\eqref{defJ} and~\eqref{defJ2}.

Note that, as a consequence of the regularity for $\tilde{g}_n$, we also obtain well posedness and differentiability with respect to the initial datum $r_1$ for the characteristics given by~\eqref{Characteristic_ODE}; in particular we are allowed
to integrate the linear equation for $h_{n+1}$ along such characteristics to obtain
\begin{equation}
h_{n+1}(u_1,r_1)=h_{0}(r_n(0))e^{\int^{u_1}_{0}G_{n|_{\chi_{n}}}dv}
-\int^{u_1}_{0}\left(G_{n}\bar{h}_{n}\right)_{|_{\chi_{n}}}e^{\int^{u_1}_{u}G_{n|_{\chi_{n}}}dv}du\;.
\label{h_{n+1}_integral}
\end{equation}
This defines a function $h_{n+1}:{\mathcal R}_{n+1}\subset[0,U]\times[0,R]\rightarrow \mathbb{R}$
where
$${\mathcal R}_{n+1}=\{(u,r) \;|\; \chi_n(u)=(u,r_n(u))=(u,r)\text{ and }r_n(0)\in[0,R]\}\;.$$
Since the problem for the characteristics is well posed, there is a characteristic through every $(u,r)\in[0,U]\times[0,R]$; in particular ${\mathcal R}_{n+1}$ is non empty, but nonetheless, integrating backwards in $u$, the characteristics may leave the fixed rectangle before reaching $u=0$, which in turn would lead to ${\mathcal R}_{n+1}\neq[0,U]\times[0,R]$. We may rule out this undesirable possibility by a choice of appropriately small $x^*$; in fact, it suffices to guarantee that the $r_n$ component of all characteristics with sufficiently large initial datum $r_1$ are nondecreasing in $u$: given $R>\sqrt{\frac{3}{\Lambda}}$, since~\eqref{X<X0} and the smallness condition on the initial data imply
$$ \|h_{n}\|_{X_{U,R}} \leq (1+C^{*})\|h_{0}\|_{X_{R}}\leq x^*\;,$$
we see that (recall~\eqref{defK})
$$K_n \leq K^* = e^{C(x^*)^2R^2}\;,$$
and from the global characterization~\eqref{charLoc}-\eqref{charCosm}
the desired monotonicity property follows if
$$r^+_{c,n}=\sqrt{\frac{K_n\Lambda}{3}}<R\;,$$
which can be arranged by choosing $x^*$ sufficiently small (see also Figure~\ref{Charact}).



We have already showed that $\bar{h}_n$, $\chi_n$ and $G_n$ have continuous partial derivatives with respect to $r$; we are then allowed to differentiate~\eqref{h_{n+1}_integral} with respect to $r_1$ and, since  $D_nh_{n+1}$ is clearly continuous, we conclude that
$$h_{n+1}\in C^1([0,U]\times [0,R])\;.$$
From the previous discussion $\partial_rDh_{n+1}$ is continuous, so differentiating equation~\eqref{h_{n+1}_integral} with respect to $r$,
and using the fact that\footnote{Here we are using the following generalized version of the Schwarz Lemma:
if $X$ and $Y$ are two nonvanishing $\mathcal{C}^1$ vector fields in $\mathbb{R}^2$ and $f$ is a $\mathcal{C}^1$ function such
that $X\cdot(Y\cdot f)$ exists and is continuous then $Y\cdot(X\cdot f)$ also exists and is equal to $X\cdot(Y\cdot f)-[X,Y]\cdot f$.}
\begin{equation*}
 \left[D_{n},\partial_{r}\right]=G_{n}\partial_{r}
\end{equation*}
we obtain the following differential equation for $\partial_{r}h_{n+1}$ (recall~\eqref{defJ}):
\begin{equation*}
\begin{aligned}
 D_{n}(\partial_{r}h_{n+1})-2G_{n}\partial_{r}h_{n+1}&=\partial_rG_{n}(h_{n+1}-\bar{h}_n)-G_n\partial_r\bar{h}_n
  \\
                                                     &=-J_{n}\frac{\partial\bar{h}_{n}}{\partial r}-\left(J_{n}-G_{n}\right)\frac{(h_{n+1}-h_{n})}{r}\;.
\end{aligned}
\end{equation*}
Using the initial conditions
\[
\partial_{r}h_{n+1}(0,r)=\partial_{r}h_{0}(r)
\]
and integrating along the characteristics leads to
\begin{equation}
\label{drhn}
 \begin{aligned}
\partial_r h_{n+1}(u_{1},r_{1})
 &=
 \partial_r h_{0}(\chi_n(0))\,e^{\int^{u_{1}}_{0}2 G_{n|_{\chi_{n}}}dv}
 \\
 &-
 \int^{u_1}_{0}\left[J_{n}\partial_r\bar{h}_{n}+\left(J_{n}-G_{n}\right)\frac{(h_{n+1}-h_{n})}{r}\right]_{|_{\chi_{n}}}e^{\int^{u_1}_{u}2 G_{n|_{\chi_{n}}}dv}du\;.
\end{aligned}
\end{equation}

By the induction hypothesis we have $G_n\leq0$ and $\|\bar h_n\|_{\our}\leq \|h_n\|_{\our}\leq \|h_0\|_{\mathcal{C}^0_R} $; therefore
\begin{equation*}
\begin{aligned}
|h_{n+1}(u_1,r_1)|
&\leq
\|h_0\|_{\mathcal{C}^0_R}\;e^{\int_0^{u_1} G_n|_{\chi_n} dv}
+\|\bar{h}_{n}\|_{\our}\; \int_0^{u_1} -G_n|_{\chi_n} e^{\int_u^{u_1} G_n|_{\chi_n}dv} du
\\
&\leq \|h_{0}\|_{\mathcal{C}^{0}_R} \underbrace{\left(e^{\int_0^{u_1} G_n|_{\chi_n} dv}
- \int_0^{u_1} G_n|_{\chi_n}  e^{\int_u^{u_1}G_n|_{\chi_n}  dv} dv\right)}_{\equiv1} =\|h_{0}\|_{\mathcal{C}^{0}_R}\; .
\end{aligned}
\end{equation*}
Then
\begin{equation}
\label{maisUmaEq}
 |h_{n+1}-h_{n}|\leq2\|h_{0}\|_{C^{0}_{R}}\quad\text{and}\quad |\partial_{r}\bar{h}_{n}|=\frac{|h_{n}-\bar{h}_{n}|}{r}\leq \frac{2\|h_{0}\|_{C^{0}_{R}}}{r}\;,
\end{equation}
so that, relying once more on Lemma~\ref{Lemma1},
\begin{equation*}
\begin{aligned}
|\partial_{r}h_{n+1}(u_1,r_1)|
&\leq
\|\partial_{r}h_{0}\|_{_{\mathcal{C}^0_R}}e^{\int^{u_1}_{0}2G_n|_{\chi_n}dv}
\\
&+
2(C_2+2C_3)\|h_{0}\|_{C^{0}_{R}}\int^{u_1}_{0}e^{\int^{u_1}_{u}2G_n|_{\chi_n}dv}du
\\
&\leq \|\partial_{r}h_{0}\|_{_{\mathcal{C}^0_R}}+2(C_2+2C_3)C_4\|h_{0}\|_{C^{0}_{R}}\;.
\end{aligned}
\end{equation*}
Setting $C^*:=2(C_2+2C_3)C_4$ it now follows that
\begin{equation*}
\begin{aligned}
\|h_{n+1}\|_{X_{U,R}}
&=\|h_{n+1}\|_{\our}+\|\partial_rh_{n+1}\|_{\our}
\\
&\leq
(1+C^{*})\|h_{0}\|_{C^{0}_{R}}+\|\partial_{r}h_{0}\|_{C^{0}_{R}}
\\
&\leq
(1+C^{*})\|h_{0}\|_{X_{R}}\;.
\end{aligned}
\end{equation*}
Thus, if $\|h_{0}\|_{X_{R}}\leq\frac{x^{*}}{1+C^{*}}$, then
\[
\|h_{n+1}\|_{X_{U,R}}\leq x^{*}
\]
and by Lemma~\ref{Lemma1}
\[
G_{n+1}\leq0\;,
\]
which completes the proof.
\end{proof}

Lemma~\ref{Lemma2} will now allow us to establish a local existence theorem for small data, while controlling the previously defined supremum norms of the solutions in terms of initial data.

\begin{thm}
\label{thm1}
Let $\Lambda>0$, $R>\sqrt{\frac{3}{\Lambda}}$ and
 $h_{0}\in\mathcal{C}^{k}([0,R])$ for $k \geq 1$.
There exists $x^{*}=x^*(\Lambda,R)>0$ and $C^{*}(x^{*},\Lambda,R)>0$, such that, if $\|h_{0}\|_{X_{R}}\leq\frac{x^{*}}{1+C^{*}}$,  then the initial value problem
\begin{equation}
\label{mainEq0}
 \left\{
\begin{array}{l}
  Dh = G\left(h-\bar{h}\right) \\
  h(0,r)= h_0(r)
\end{array}
\right.
\end{equation}
has a unique solution $h\in\mathcal{C}^{k}([0,U]\times[0,R])$, for $U=U(x^*/(1+C^*);R,\Lambda)$ sufficiently small. Moreover,
\begin{equation}
\label{boundh1}
\left\|h\right\|_{\mathcal{C}^{0}_{U,R}}=\left\|h_{0}\right\|_{\mathcal{C}^{0}_{R}}
\end{equation}
and
\begin{equation}
\label{boundh2}
\left\|h\right\|_{X_{U,R}}\leq (1+C^*)\,\left\|h_{0}\right\|_{X_{R}}.
\end{equation}
\end{thm}

\begin{proof}

Fix $x^*$ as in Lemma~\ref{Lemma1} and consider a sequence $\{h_n\}$ as defined in Lemma~\ref{Lemma2}, with $U< 1$. From~\eqref{Est_h_bar_h} and Lemma~\ref{Lemma2} we have
\begin{equation*}
\begin{aligned}
 \left|(h_{n}-\bar{h}_{n})+(h_{n-1}-\bar{h}_{n-1})\right|&\leq
\frac{r}{2}\left(\left\|\partial_{r}h_{n}\right\|_{\mathcal{C}^{0}_{U,R}}+\left\|\partial_{r}h_{n-1}\right\|_{\mathcal{C}^{0}_{U,R}}\right)
\\
&\leq (1+C^{*})\,r\,\|h_{0}\|_{X_R}\leq x^*r\;,
\end{aligned}
\end{equation*}
and
\begin{equation*}
  \left|(h_{n}-\bar{h}_{n})-({h}_{n-1}-\bar{h}_{n-1})\right|=\left|(h_{n}-h_{n-1})-\overline{({h}_{n}-{h}_{n-1})}\right|\leq 2\left\|h_{n}-h_{n-1}\right\|_{\mathcal{C}^{0}_{U,R}}
\end{equation*}
so that
\begin{equation}
\label{integranda}
\begin{aligned}
\left|(h_{n}-\bar{h}_{n})^2-(h_{n-1}-\bar{h}_{n-1})^2\right|
&\leq\left|(h_{n}-\bar{h}_{n})+(h_{n-1}-\bar{h}_{n-1})\right|\,\left|(h_{n}-\bar{h}_{n})-({h}_{n-1}-\bar{h}_{n-1})\right|
\\
&\leq 2x^*\,r\left\|h_{n}-h_{n-1}\right\|_{\mathcal{C}^{0}_{U,R}}
\\
&= C\,r \left\|h_{n}-h_{n-1}\right\|_{\mathcal{C}^{0}_{U,R}} \;
\end{aligned}
\end{equation}
(we will, until the end of this proof, allow the constants to depend on $x^*$, besides the fixed parameters $\Lambda$ and $R$).
The mean value theorem yields the following elementary inequality
\begin{equation}
\label{expIneq}
|e^x-e^y|\leq \max\{e^x,e^y\}|x-y|\;,
\end{equation}
from which (recall~\eqref{g*})
\begin{equation}
\label{deltag}
\begin{aligned}
|g_{n}-g_{n-1}|
&=
\left|\exp\left(C\int_0^{r}\frac{(h_{n}-\bar{h}_{n})^2}{s}\right)-\exp\left(C\int_0^{r}\frac{(h_{n-1}-\bar{h}_{n-1})^2}{s}\right)\right|
\\
&\lesssim
K^*\int^{r}_{0}\frac{\left|(h_{n}-\bar{h}_{n})^2-(h_{n-1}-\bar{h}_{n-1})^2\right|}{s}ds
\\
&\leq
C\,r\,\left\|h_{n}-h_{n-1}\right\|_{\mathcal{C}^{0}_{U,R}}\,.
\end{aligned}
\end{equation}
Then
\begin{equation}
\label{deltaTildeg}
\begin{aligned}
|\tilde{g}_{n}-\tilde{g}_{n-1}|
&=
\left|\frac{1}{r}\int^{r}_{0}(g_{n}-g_{n-1})(1-\Lambda s^2)ds\right|
\\
&\leq
C \,r\,\left\|h_{n}-h_{n-1}\right\|_{\mathcal{C}^{0}_{U,R}}\,,
\end{aligned}
\end{equation}
and using~\eqref{defG2},
\begin{equation*}
\begin{aligned}
|G_{n}-G_{n-1}|
&=
\frac{1}{2r}\left|(g_{n}-g_{n-1})(1-\Lambda r^2)-\frac{1}{r}\int^{r}_{0}(g_{n}-g_{n-1})(1-\Lambda s^2)ds\right|
\\
&\leq
C\,\left\|h_{n}-h_{n-1}\right\|_{\mathcal{C}^{0}_{U,R}}.
\end{aligned}
\end{equation*}
Note that, since $r\leq R$, the $r$ factors in the previous estimates may be absorbed by the corresponding constants.

Until now we have been estimating the difference between consecutive terms of sequences with both terms evaluated at the same  point $(u,r)$, but we will also need to estimate differences between consecutive terms evaluated at the corresponding characteristics; more precisely, for a given sequence $f_n$ we will estimate
\begin{equation*}
\begin{aligned}
|f_{n|{\chi_n}}-f_{n-1|{\chi_{n-1}}}|
&=
|f_n(u,r_n(u))-f_{n-1}(u,r_{n-1}(u))|
\\
&\leq |f_n(u,r_n(u))-f_{n}(u,r_{n-1}(u))|+|f_n(u,r_{n-1}(u))-f_{n-1}(u,r_{n-1}(u))|\;.
\end{aligned}
\end{equation*}
If for the second term we have, as before, a uniform estimate of the form $C\|h_n-h_{n-1}\|_{\our}$, and for the first one of the form $C|r_n-r_{n-1}|$, then, by~\eqref{deltaChiFinal} below, we will obtain, since $u_1\leq U<1$,
\begin{equation}
\label{deltaFFinal}
|f_{n|{\chi_n}}-f_{n-1|{\chi_{n-1}}}|\leq C\|f_n-f_{n-1}\|_{\our}\;.
\end{equation}
Also, if $\|\partial_rf_n\|_{\our}\leq C$ then the desired
$$|f_n(u,r_2)-f_n(u,r_1)|\leq\left| \int_{r_1}^{r_2}\partial_rf_n(r)dr\right|\leq C|r_2-r_1|\;,$$
follows immediately. We have (see~\eqref{Est_h_bar_h})
$$|\partial_r\bar{h}_n|=\left|\frac{h_n-\bar{h}_n}{r}\right|\leq C\;,$$
and from\eqref{Est_partial_g}
$$|\partial_r g_n|\leq C r\;.$$
By Lemma~\ref{Lemma1} we have $\|G_n\|_{\our}\leq C$, which in view of \eqref{defG} is equivalent to $\|\partial_r\tilde g_n\|_{\our}\leq C$;  since~\eqref{defJ},~\eqref{defJ2} and~\eqref{Gbound2} together with the above bounds yield $\|\partial_rG_n\|_{\our}\leq C$, the desired estimates, of the form~\eqref{deltaFFinal}, follow for the sequences $h_n$, $\bar{h}_n$, $g_n$, $\tilde g_n$ and $G_n$ once we have proved~\eqref{deltaChiFinal}. To do this, start from equation~\eqref{Characteristic_ODE} for the characteristics to obtain
$$r_{n}(u)=r_n(u_1)+\frac{1}{2}\int_u^{u_1}\tilde g_{n}(s,r_n(s))ds\;,$$
so that the difference between two consecutive characteristics through $(u_1,r_1)$ satisfies
\begin{equation*}
\begin{aligned}
r_{n}(u)-r_{n-1}(u)
&= \frac{1}{2}\int_u^{u_1}\left\{\tilde g_{n}(s,r_n(s))-g_{n-1}(s,r_{n-1}(s))\right\}ds
\\
&=
\frac{1}{2}\int_u^{u_1}\left\{\tilde g_{n}(s,r_n(s))-\tilde g_{n}(s,r_{n-1}(s))\right\}ds
+
\frac{1}{2}\int_u^{u_1}\left\{\tilde g_{n}(s,r_{n-1}(s))-\tilde g_{n-1}(s,r_{n-1}(s))\right\}ds\;.
\end{aligned}
\end{equation*}
From the previously obtained bounds $\|\partial_r\tilde g_n\|_{\our}\leq C$ and~\eqref{deltaTildeg}, we then have
$$|r_{n}(u)-r_{n-1}(u)|\leq C\int_u^{u_1}|r_{n}(s)-r_{n-1}(s)|ds + C'(u_1-u)\|h_n-h_{n-1}\|_{\our}\;,$$
from which\footnote{Here we used the following comparison principle: if $y,z\in\mathcal{C}^0([t_0,t_1])$ satisfy $y(t)\leq f(t)+C\int_t^{t_1}y(s)ds$ and $z(t)= f(t)+C\int_t^{t_1}z(s)ds$, then $y(t)\leq z(t)$, $\forall t\in[t_0,t_1]$\;.}
\begin{equation}
\label{deltaChiFinal}
|r_{n}(u)-r_{n-1}(u)|\leq \frac{C'}{C}\|h_n-h_{n-1}\|_{\our}\left(e^{C(u_1-u)}-1\right)\;,
\end{equation}
as desired.

Now, from~\eqref{h_{n+1}_integral} and the elementary identity
$$a_2b_2c_2-a_1b_1c_1=(a_2-a_1)b_2c_2+(b_2-b_1)a_1c_2+(c_2-c_1)a_1b_1$$
%
we get
\begin{equation*}
\begin{aligned}
|(h_{n+1}-h_n)(u_1,r_1)|
&\leq
\|h_0\|_{\mathcal{C}^0_R}\underbrace{\left|\exp\left(\int_0^{u_1}G_{n|\chi_n}dv\right)-\exp\left(\int_0^{u_1}G_{n-1|\chi_{n-1}}dv\right)\right|}_{I}
\\
&+\underbrace{\int_0^{u_1}\left|G_{n|\chi_n}-G_{n-1|\chi_{n-1}}\right|\left|\bar{h}_{n|\chi_n}\right|\exp\left(\int_u^{u_1}G_{n|\chi_n}dv\right)du}_{II}
\\
&+\underbrace{\int_0^{u_1}\left|\bar h_{n|\chi_n}-\bar h_{n-1|\chi_{n-1}}\right|\left|G_{n-1|\chi_{n-1}}\right|\exp\left(\int_u^{u_1}G_{n|\chi_n}dv\right)du}_{III}
\\
&+\underbrace{\int_0^{u_1}\left|\exp\left(\int_u^{u_1}G_{n|\chi_n}dv\right)-\exp\left(\int_u^{u_1}G_{n-1|\chi_{n-1}}dv\right)\right|\left|G_{n-1|\chi_{n-1}}\bar{h}_{n-1|\chi_{n-1}}\right|du}_{IV}\;.
\end{aligned}
\end{equation*}
Using~\eqref{G<0},~\eqref{expIneq} and~\eqref{deltaFFinal}, which holds for the sequence $G_n$ as discussed earlier, gives
$$I\leq \left|\left(\int_0^{u_1}G_{n|\chi_n}dv\right)-\left(\int_0^{u_1}G_{n-1|\chi_{n-1}}dv\right)\right|\leq C u_1\left\|h_{n}-h_{n-1}\right\|_{\mathcal{C}^{0}_{U,R}}\;,$$
and, in view also of~\eqref{h<h0},
$$II\leq \int_0^{u_1}C\|h_0\|_{\mathcal{C}^{0}_R} \left\|h_{n}-h_{n-1}\right\|_{\our}du\leq Cu_1\left\|h_{n}-h_{n-1}\right\|_{\our}\;.$$
In a similar way (recall that~\eqref{deltaFFinal} also holds for the sequence $\bar{h}_n$)
$$III\leq Cu_1\left\|h_{n}-h_{n-1}\right\|_{\mathcal{C}^{0}_{U,R}}\;,$$
and, using the bound for $I$,
$$IV\leq C u_1^2\left\|h_{n}-h_{n-1}\right\|_{\mathcal{C}^{0}_{U,R}}\;.$$
Putting all the pieces together yields (recall that we have imposed the restriction $u_1\leq U<1$)
\begin{equation}
\label{contractionI}
\|h_{n+1}-h_n\|_{\our}\leq C\,U\|h_{n}-h_{n-1}\|_{\our}.
\end{equation}

Now, applying the same strategy to~\eqref{drhn} leads to
\begin{equation*}
\begin{aligned}
|(\partial_rh_{n+1}-\partial_rh_n)(u_1,r_1)|
&\leq
\|\partial_rh_0\|_{\mathcal{C}^0_R}\underbrace{\left|\exp\left(\int_0^{u_1}2G_{n|\chi_n}dv\right)-\exp\left(\int_0^{u_1}2G_{n-1|\chi_{n-1}}dv\right)\right|}_{(i)}
\\
&+\underbrace{\int_0^{u_1}\left|J_{n|\chi_n}-J_{n-1|\chi_{n-1}}\right|\left|\partial_r\bar{h}_{n|\chi_n}\right|\exp\left(\int_u^{u_1}2G_{n|\chi_n}dv\right)du}_{(ii)}
\\
&+\underbrace{\int_0^{u_1}\left|\partial_r\bar h_{n|\chi_n}-\partial_r\bar h_{n-1|\chi_{n-1}}\right|\left|J_{n-1|\chi_{n-1}}\right|\exp\left(\int_u^{u_1}2G_{n|\chi_n}dv\right)du}_{(iii)}
\\
&+\underbrace{\int_0^{u_1}\left|\exp\left(\int_u^{u_1}2G_{n|\chi_n}dv\right)-\exp\left(\int_u^{u_1}2G_{n-1|\chi_{n-1}}dv\right)\right|\left|J_{n-1|\chi_{n-1}}\partial_r\bar{h}_{n-1|\chi_{n-1}}\right|du}_{(iv)}
\\
&+\underbrace{\int^{u_1}_{0}\left[\left|J_{n}-G_{n}\right|\frac{|h_{n+1}-h_{n}|}{r}\right]_{|\chi_n}e^{\int^{u_1}_{u}2 G_{n|\chi_n}dv}du}_{(v)}
\\
&+\underbrace{\int^{u_1}_{0}\left[\left|J_{n-1}-G_{n-1}\right|\frac{|h_{n}-h_{n-1}|}{r}\right]_{|\chi_{n-1}}e^{\int^{u_1}_{u}2 G_{n-1|\chi_{n-1}}dv}du}_{(vi)}\;.
\end{aligned}
\end{equation*}

We have
\begin{equation*}
\begin{aligned}
|\partial_r g_{n}-\partial_r g_{n-1}|
&\lesssim \left|g_n\frac{(h_{n}-\bar{h}_{n})^2}{r}-g_{n-1}\frac{(h_{n-1}-\bar{h}_{n-1})^2}{r}\right|
\\
&\lesssim
|g_n|\frac{\left|(h_{n}-\bar{h}_{n})^2-(h_{n-1}-\bar{h}_{n-1})^2\right|}{r}
+|g_n-g_{n-1}|\frac{(h_{n-1}-\bar{h}_{n-1})^2}{r}
\\
&\lesssim \|h_{n}-h_{n-1}\|_{\our}\;,
\end{aligned}
\end{equation*}
where we have used~\eqref{g*},~\eqref{integranda} and~\eqref{deltag}.
Similarly
\begin{equation*}
\begin{aligned}
|\partial_r g_{n}(u,r_2)-\partial_r g_{n}(u,r_1)|
&\lesssim
|g_n(u,r_2)| \left|\frac{(h_{n}-\bar{h}_{n})^2(u,r_2)}{r_2} - \frac{(h_{n}-\bar{h}_{n})^2(u,r_1)}{r_1} \right|+
\\
&+|g_n(u,r_2)-g_{n}(u,r_1)|\left|\frac{(h_{n-1}-\bar{h}_{n-1})^2(u,r_1)}{r_1}\right|
\\
&\lesssim |r_2-r_1|\;.
\end{aligned}
\end{equation*}
We conclude that~\eqref{deltaFFinal} holds for the sequence $\partial_rg_n$ and since it also holds for the sequences $g_n$ and $G_n$ we obtain from~\eqref{defJ2}
\[
\left|J_{n|\chi_{n}}-J_{n-1|\chi_{n-1}}\right|\leq C \|h_{n}-h_{n-1}\|_{\our}\;.
\]
As an immediate consequence one obtains for $(i)-(iv)$ estimates similar to the ones derived for $I-IV$ (recall~\eqref{meanDif},~\eqref{Est_h_bar_h} and~\eqref{boundJ}).

Using~\eqref{Gbound2} and~\eqref{boundJ} we also have
$$(vi)\leq CU \|h_{n}-h_{n-1}\|_{\our}\,, $$
and~\eqref{contractionI} provides
$$(v)\leq CU^2 \|h_{n}-h_{n-1}\|_{\our}\,. $$

We finally obtain
\begin{equation*}
\begin{aligned}
\|h_{n+1}-h_n\|_{\Xour}
&=\|h_{n+1}-h_n\|_{\our}+\|\partial_rh_{n+1}-\partial_rh_n\|_{\our}
\\
&\leq CU \|h_{n}-h_{n-1}\|_{\our}
\\
&\leq CU \|h_{n}-h_{n-1}\|_{\Xour}\;.
\end{aligned}
\end{equation*}
So, for  $U$ sufficiently small, $\{h_n\}$ contracts, and consequently converges, with respect to $\|\cdot\|_{\Xour}$. The previous estimates show that the convergence of $h_n$ lead to the uniform convergence of all the sequences appearing in~\eqref{h_{n+1}_integral} and~\eqref{drhn}. Taking the limit of~\eqref{h_{n+1}_integral} leads to
\begin{equation}
 h(u_1,r_1)=h_{0}(\chi(0))e^{\int_0^{u_1}G|_{\chi}dv}-\int_0^{u_1}\left(G\bar{h}\right)|_{\chi}e^{\int^{u_1}_{u}G|_{\chi}dv}du\,,
\label{solutionIntegral}
\end{equation}
where we denote the limiting functions by removing the indices. Equation~\eqref{solutionIntegral} shows that $h$ is a continuous solution to~\eqref{mainEq0}, the limit of~\eqref{drhn} shows that $\partial_rh$ solves~\eqref{D_partial_h} and is continuous, and we see that $h\in\mathcal{C}^1$, since $Dh$ is also clearly continuous.

Now let $1 \leq m < k$ be an integer, and assume that $h\in\mathcal{C}^m$. As in the proof of Lemma~\ref{Lemma2}, but using the Taylor expansion of order $m$, we can show that $\bar{h}$, $\tilde{g}$ (which controls the characteristics), $G$ and $\partial_rG$ are also $\mathcal{C}^m$, from which it follows that $\partial_r(G\bar{h})$ is $\mathcal{C}^m$. Taking the partial derivatives of~\eqref{solutionIntegral} as in~\cite{Costa:2012} (using the assumed regularity of the initial data) we then see that actually $h\in\mathcal{C}^{m+1}$, and so $h\in\mathcal{C}^{k}$.   

\vspace{0,5cm}

To establish uniqueness consider two solutions of~\eqref{mainEq0} and derive the following evolution equation for their difference:
\begin{equation}
\label{eqDiference}
D_1(h_2-h_1)-G_1(h_2-h_1)=\frac{1}{2}\left(\tilde g_2-\tilde g_1\right)\partial_rh_2+\left(G_2-G_1\right)\left({h}_2-\bar{h}_2\right)
-G_1\left(\bar{h}_2-\bar{h}_1\right)\;.
\end{equation}
Integrating it along the characteristics associated to $h_1$ yields
$$|(h_2-h_1)(u_1,r_1)|\leq\int_0^{u_1}\left[\frac{1}{2}\left|\tilde g_2-\tilde g_1\right||\partial_rh_2|+\left|G_2-G_1\right|\left|{h}_2-\bar{h}_2\right|
+|G_1|\left|\bar{h}_2-\bar{h}_1\right|\right]_{|\chi_1}e^{\int_u^{u_1}G_{1|\chi_1}dv}du\;.$$
Setting
$$\delta(u)=\|(h_2-h_1)(u,\cdot)\|_{{\mathcal C}^0_R}\;,$$
then, arguing as in the beginning of the proof of this theorem, we obtain, from the previous inequality,
$$\delta(u_1)\leq C\int_0^{u_1}\delta(u)e^{\int_u^{u_1}G_{1|\chi_1}dv}du\;.$$
Applying Gronwall's inequality we conclude that
$$\delta(u)\leq 0\;,$$
and uniqueness follows.
\vspace{0,5cm}

The estimates~\eqref{boundh1} and~\eqref{boundh2} are now an immediate consequence of Lemma~\ref{Lemma2}.
\end{proof}


%


\section{Global existence in time}

\begin{thm}
\label{thmGlobal}
Let $\Lambda>0$, $R>\sqrt{\frac{3}{\Lambda}}$ and
$h_{0}\in\mathcal{C}^{k}([0,R])$ for $k \geq 1$.
There exists $x^{*}=x^*(\Lambda,R)>0$ and $C^{*}(x^{*},\Lambda,R)>0$, such that, if $\|h_{0}\|_{X_{R}}\leq\frac{x^{*}}{(1+C^{*})^2}$,  then the initial value problem
\begin{equation}
\label{mainEq}
 \left\{
\begin{array}{l}
  Dh = G\left(h-\bar{h}\right) \\
  h(0,r)= h_0(r)
\end{array}
\right.
\end{equation}
has a unique solution $h\in\mathcal{C}^{k}([0,\infty]\times[0,R])$. Moreover,
\begin{equation}
\label{boundh3}
\|h\|_{{\mathcal C}^0([0,\infty)\times[0,R])} = \|h_{0}\|_{{\mathcal C}^0([0,R])}\;,
\end{equation}
and
\begin{equation}
\label{boundh4}
\|h\|_{{X}([0,\infty)\times[0,R])} \leq (1+C^*) \|h_{0}\|_{{X}([0,R])}\;.
\end{equation}
%
Also, solutions depend continuously on initial data in the following precise sense: if $h^1$ and $h^2$ are two solutions with initial data $h_0^1$ and $h_0^2$, respectively, then
\[
\|h^1-h^2\|_{\our} \leq C(U,R,\Lambda) \|h_0^1-h_0^2\|_{{\mathcal C}^0_R}\,,
\]
for all $U>0$\;.
\end{thm}

\begin{proof}
From Theorem~\ref{thm1} there exists a unique $h^1\in\mathcal{C}^k([0,U_{1}]\times[0,R])$ solving~\eqref{mainEq},
with existence time $U_1=U(x^*/(1+C^*)^2)$\;. Moreover
\[
\|h^{1}(U_1,\cdot)\|_{X_{R}}\leq \|h^{1}\|_{X_{U_1,R}} \leq (1+C^{*}) \|h_{0}\|_{X_R}\leq\frac{x^*}{1+C^*}\;.
\]
So Theorem~\ref{thm1} provides a solution $h^2\in\mathcal{C}^k([0,U_{2}]\times[0,R])$ with initial data $h^2(0,r) = h^{1}(U_1,r)$ and existence time
$U_2=U(x^*/(1+C^*))$.
Now,
\[
 h:[0,U_{1}+U_{2}]\times[0,R]\rightarrow\mathbb{R}
\]
defined by
\begin{equation*}
 h(u,r):=\left\{
\begin{array}{l}
  h^{1}(u,r)\quad,\quad u\in[0,U_1]  \\
  h^{2}(u,r)\quad,\quad u\in[U_1, U_{1}+U_{2}]\;.
\end{array}
\right.
\end{equation*}
is the unique solution of our problem in $\mathcal{C}^k([0,U_1+U_2]\times[0,R])$.
Since~\eqref{boundh1} applies to both $h^1$ and $h^2$ we see that:
\[
\|h^1\|_{\mathcal{C}^{0}_{U_1,R}}=\|h_{0}\|_{\mathcal{C}^{0}_{R}}\;,
\]
so that
\[
\|h^2\|_{\mathcal{C}^{0}_{U_2,R}}=\|h^1(U_1,\cdot)\|_{\mathcal{C}^{0}_{R}}\leq\|h_{0}\|_{\mathcal{C}^{0}_{R}}\;,
\]
and hence
\begin{equation}
\label{bound+1}
\|h\|_{\mathcal{C}^{0}_{U_1+U_2,R}}=\|h_{0}\|_{\mathcal{C}^{0}_{R}}\;.
\end{equation}
Arguing as in the proof of Lemma~\ref{Lemma2}, we see that $\partial_rDh$  is continuous and consequently $\partial_rh$ solves~\eqref{D_partial_h} so that:
\begin{equation}
\label{drh}
\partial_r h(u_{1},r_{1})=
\partial_r h_0(\chi(0))\,e^{\int_0^{u_{1}}2 G_{|\chi}dv}
-\int_0^{u_1}\left(J\partial_r\bar{h}\right)_{|\chi}e^{\int^{u_1}_{u}2 G_{|{\chi}}dv}du\;.
\end{equation}
Consequently,
\begin{equation*}
\begin{aligned}
|\partial_{r}h(u_1,r_1)|
&\leq
|\partial_{r}h_{0}(r_{0})|e^{\int^{u_{1}}_{0}2Gdv}+\int^{u_{1}}_{0}|J||\partial_r\bar{h}|e^{\int^{u_{1}}_{u}2Gdv}du
\\
&\leq \|\partial_rh_0\|_{_{\mathcal{C}^{0}_{R}}}+2C_3C_4\|h_0\|_{_{\mathcal{C}^{0}_{R}}}
\\
&\leq \|\partial_rh_0\|_{_{\mathcal{C}^{0}_{R}}}+C^*\|h_0\|_{_{\mathcal{C}^{0}_{R}}}\;,
\end{aligned}
\end{equation*}
where we have used an estimate analogous to~\eqref{maisUmaEq}, the fact that Lemma~\ref{Lemma1} applies to $h$ (with the same notation for the constants), and the fact that we may choose $C^*:=2(2C_2+C_3)C_4$, which can be traced back to the proof of Lemma~\ref{Lemma2}.

Combining the last two estimates with the smallness condition on the initial data leads to:
\begin{equation}
\label{criterioExtensao}
\|h\|_{X_{U_1+U_2,R}}\leq (1+C^*)\|h_0\|_{X_R}\leq \frac{x^*}{1+C^*}\;.
\end{equation}
So, by Theorem~\ref{thm1}, we can extend the solution by the same amount $U_2=U(x^*/(1+C^*))$ as before; the global (in time) existence then follows, with the bounds~\eqref{boundh3} and~\eqref{boundh4} a consequence of~\eqref{bound+1} and~\eqref{criterioExtensao}.

The continuous dependence statement follows by applying Gronwall's inequality to the integral inequality obtained integrating equation~\eqref{eqDiference} and using the estimates derived in the  beginning of the proof of Theorem~\ref{thm1}.

\end{proof}

\section{Exponential decay}

\begin{thm} \label{thmDecay}
Let $\Lambda>0$, $R>\sqrt{\frac{3}{\Lambda}}$ and set $H=2\sqrt{\frac{\Lambda}{3}}$. Then, for $\|h_0\|_{X_R}$ sufficiently small, the solution, $h\in {\mathcal C}^k([0,\infty]\times[0,R])$, of~\eqref{mainEq} satisfies
\[
\sup_{0\leq r\leq R}\left|\partial_rh(u,r)\right|\leq \hat{C} e^{-Hu}\;,
\]
and, consequently, there exits $\underline{h}\in\mathbb{R}$ such that
\[
|h(u,r)-\underline{h}|\leq \bar{C} e^{-Hu}\;,
\]
with constants $\hat{C}$ and $\bar{C}$ depending on $\|h_0\|_{X_R}$, $R$ and $\Lambda$\;.
\end{thm}

\begin{proof}
Consider the solution provided by Theorem~\ref{thmGlobal}. Set
\begin{equation}
\label{energyDef}
{\mathcal E}(u):=\|\partial_rh(u,\cdot)\|_{{\mathcal C}^0_R}\;,
\end{equation}
and
\[
E(u_0):=\sup_{u\geq u_0} {\mathcal E}(u)\;.
\]
Arguing as in~\eqref{Est_h_bar_h} we get
\begin{equation}
\label{Est_h_bar_hV2}
|(h-\bar{h})(u,r)|\leq \frac{r}{2}{\mathcal E}(u)\;.
\end{equation}
Lemma~\ref{Lemma1} applies and note that, for a fixed $x_0\geq0$, the estimates~\eqref{Gbound1},~\eqref{Gbound2} and~\eqref{boundJ} are still valid, with $x^*$ replaced with $E(u_0)$, for the functions $G$ and $J$ restricted to $[u_0,\infty)\times [0,R]$. Integrating~\eqref{D_partial_h} with initial data on $u=u_0$ gives, for $u_1\geq u_0$ (compare with~\eqref{drh})
\[
{\mathcal E}(u_1)\leq {\mathcal E}(u_0)e^{-2C_1\int_{u_0}^{u_1}r(s)ds}+\frac{C_3R}{2}\int_{u_0}^{u_1} {\mathcal E}(u)e^{-2C_1\int_u^{u_1}r(v)dv}du\;,
\]
by using~\eqref{Est_h_bar_hV2} and~\eqref{Definitionh}; once again we have used the notation for the constants set by Lemma~\ref{Lemma1}.

Recall that $r(u)=r(u;u_1,r_1)$ and that if $r_1< r^{-}_{c}=\sqrt{\frac{3}{\Lambda {\myK}}}$ then, as in the calculations leading to~\eqref{intLocalRegion}, we have
\[
{\mathcal E}(u_1)\leq {\mathcal E}(u_0)2^{C_1/\alpha^2}e^{-\frac{C_1}{\alpha}{u_1}}+2^{C_1/\alpha^2-1}C_3R\int_{u_0}^{u_1} {\mathcal E}(u)e^{\frac{C_1}{\alpha}(u-u_1)}du\;,
\]
so that
\[
e^{\frac{C_1}{\alpha}{u_1}}{\mathcal E}(u_1)\leq 2^{C_1/\alpha^2}{\mathcal E}({u_0})+2^{C_1/\alpha^2-1}C_3R\int_{u_0}^{u_1} {\mathcal E}(u)e^{\frac{C_1}{\alpha}u}du\;.
\]
Applying Gronwall's Lemma to ${\mathcal F}(u_1):=e^{\frac{C_1}{\alpha}{u_1}}{\mathcal E}(u_1)$ then gives
\[
e^{\frac{C_1}{\alpha}{u_1}}{\mathcal E}(u_1)\leq 2^{C_1/\alpha^2}{\mathcal E}({u_0})\exp\left(2^{C_1/\alpha^2-1}C_3R (u_1-u_0)\right)\;,
\]
so that finally
\[
{\mathcal E}(u_1)\leq 2^{C_1/\alpha^2}{\mathcal E}({u_0})\exp\left\{\left(2^{C_1/\alpha^2-1}C_3R-\frac{C_1}{\alpha} \right)u_1\right\}\;.
\]
For $r_1\geq r^{-}_{c}$ we have~\eqref{Est_char_cosm} instead and a similar, although simpler, derivation yields
\[
{\mathcal E}(u_1)\leq {\mathcal E}({u_0})\exp\left\{\left(\frac{C_3R}{2}-{2C_1}r^{-}_{c} \right)u_1\right\}\;.
\]
Observe that $K=e^{O(E(u_0))}$, $C_1=\frac{\Lambda}{3}+O(E(u_0))$, $C_3=O(E(u_0))$, uniformly in $u_0$ since $u_0\mapsto E(u_0)$ is bounded. Using such boundedness once more,  we can encode the previous estimates into
\begin{equation}
\label{energyIneq5}
{\mathcal E}(u)\leq C e^{-\hat{H}(u_0)u}\;,
\end{equation}
with
\begin{equation}
\label{cH}
 \hat{H}(u_0)=H+O(E(u_0))\;.
\end{equation}
Since $E(u_0)$ is controlled by $\|h_0\|_{X_R}$ (see \eqref{boundh4}), choosing the later sufficiently small leads to 
%
%
%
%
\[
\hat{H}(u_0)\geq \mathring{H}>0\;,
\]
so that~\eqref{energyIneq5} implies
\begin{equation*}
\begin{aligned}
{\mathcal E}(u)
&\leq
C e^{-\mathring{H}u}
\end{aligned}
\end{equation*}
for $u \geq u_0$. Then clearly
$$
E(u_0)\leq C e^{-\mathring{H}u_0}\;,
$$
so that~\eqref{cH} becomes
$$|H-\hat{H}(u_0)|\leq C e^{-\mathring{H}u_0}\;. $$
Finally, setting $u_0=\frac{u}2$ yields
\begin{equation*}
\begin{aligned}
e^{Hu}{\mathcal E}(u)
&\leq
 C\exp(Hu-\hat{H}(u/2)u)
\\
&\leq
C\exp(Ce^{-\mathring{H}u/2}u)\leq \hat{C}\;,
\end{aligned}
\end{equation*}
as desired; the remaining claims follow as in \cite{Costa:2012}.
\end{proof}

It is now clear from \eqref{meanDef}, \eqref{g} and \eqref{barg_terms_g} that
\begin{align}
& |\bar{h}(u,r)-\underline{h}|\leq \bar{C} e^{-Hu}\;, \label{asympt1} \\
& |g-1|\leq \bar{C} e^{-Hu}\;, \label{asympt2} \\
& |\tilde{g}-1+\Lambda r^2 / 3|\leq \bar{C} e^{-Hu}\;. \label{asympt3}
\end{align}
In particular, \eqref{massfunction} implies that
\[
m(u) \leq \bar{C} e^{-Hu}\;,
\]
and so the final Bondi mass $M_1$ vanishes. Finally, geodesic completeness is easily obtained from \eqref{asympt1}-\eqref{asympt3}.

\section*{Acknowledgements}

We thank Pedro Gir\~ao, Marc Mars, Alan Rendall, Jorge Silva and Ra\"ul Vera for useful discussions. This work was supported by projects PTDC/MAT/108921/2008 and CERN/FP/116377/2010, and by CMAT, Universidade do Minho, and CAMSDG, Instituto Superior T\'ecnico, through FCT plurianual funding. AA thanks the Mathematics Department of Instituto Superior T\'ecnico (Lisbon), where this work was done, and the International Erwin Schr\"odinger Institute (Vienna), where the workshop ``Dynamics of General Relativity: Analytical and Numerical Approaches'' took place, for hospitality, and FCT for grant SFRH/BD/48658/2008.

\section*{Appendix: the Einstein equations}
For the metric \eqref{metricBondi} we have the nonvanishing
\begin{itemize}

 \item (inverse) metric components:

\begin{equation*}
\begin{aligned}
 &g_{uu}=-g\tilde{g}\;,\quad g_{ur}=g_{ru}=-g\;,\quad g_{\theta\theta}=r^{2}\;,\quad g_{\varphi\varphi}=r^{2}\sin^{2}{\theta}=\sin^{2}{\theta}g_{\theta\theta}\;, \\
 &g^{rr}=\frac{\tilde{g}}{g}\;,\quad g^{ru}=g^{ur}=-\frac{1}{g}\;,\quad g^{\theta\theta}=\frac{1}{r^{2}}\;,\quad g^{\varphi\varphi}=\frac{1}{r^{2}\sin^{2}{\theta}}=\frac{g^{\theta\theta}}{\sin^{2}{\theta}}\;;
\end{aligned}
\end{equation*}

\item Christoffel symbols $\Gamma^{\mu}_{\alpha\beta}=\frac{1}{2}g^{\mu\nu}\left(\partial_{\alpha}g_{\beta\nu}+\partial_{\beta}g_{\nu\alpha}-\partial_{\nu}g_{\alpha\beta}\right)$:

\begin{equation*}
\begin{aligned}
 &\Gamma^{u}_{uu}=\frac{1}{g}\left[\frac{\partial}{\partial u}g-\frac{1}{2}\frac{\partial}{\partial r}(g\tilde{g})\right]\;,\quad \Gamma^{u}_{\theta\theta}=\frac{r}{g}\;,\quad \Gamma^{u}_{\varphi\varphi}=\frac{r\sin^{2}{\theta}}{g}=\sin^{2}{\theta}\Gamma^{u}_{\theta\theta}\;, \\
 &\Gamma^{r}_{uu}=\frac{1}{2g}\frac{\partial}{\partial u}(g\tilde{g})-\frac{\tilde{g}}{g}\left[\frac{\partial}{\partial u}g-\frac{1}{2}\frac{\partial}{\partial r}(g\tilde{g})\right]\;,\quad \Gamma^{r}_{ur}=\frac{1}{2g}\frac{\partial}{\partial r}(g\tilde{g})\;,\quad \Gamma^{r}_{rr}=\frac{1}{g}\frac{\partial}{\partial r}g\;,  \\
 &\Gamma^{r}_{\theta\theta}=-\frac{\tilde{g}}{g}r\;,\quad  \Gamma^{r}_{\varphi\varphi}=-\frac{\tilde{g}}{g}r\sin^{2}{\theta}=\sin^{2}{\theta}\Gamma^{r}_{\theta\theta}\;,  \\
 &\Gamma^{\theta}_{\varphi\varphi}=-\sin{\theta}\cos{\theta}\;,\quad \Gamma^{\theta}_{\theta r}=\frac{1}{r}=\Gamma^{\varphi}_{\varphi r}\;,\quad \Gamma^{\varphi}_{\theta\varphi}=\frac{\cos{\theta}}{\sin{\theta}}=-\frac{1}{\sin^{2}{\theta}}\Gamma^{\theta}_{\varphi\varphi}\;;
\end{aligned}
\end{equation*}

\item Ricci tensor components $R_{\alpha\beta}=\partial_{\mu}\Gamma^{\mu}_{\alpha\beta}-\partial_{\alpha}\Gamma^{\mu}_{\mu\beta}+\Gamma^{\nu}_{\alpha\beta}\Gamma^{\mu}_{\nu\mu}-\Gamma^{\nu}_{\mu\beta}\Gamma^{\mu}_{\nu\alpha}$:

\begin{equation*}
 \begin{aligned}
  R_{uu}&=\partial_{r}\Gamma^{r}_{uu}-\partial_{u}\Gamma^{r}_{ur}+\Gamma^{r}_{uu}\left(\Gamma^{r}_{rr}+\Gamma^{\theta}_{\theta r}+\Gamma^{\varphi}_{\varphi r}\right)+\Gamma^{r}_{ru}\left(\Gamma^{u}_{uu}-\Gamma^{r}_{ru}\right) \\
       &=\frac{1}{2g}\left(\frac{\partial^{2}}{\partial r\partial u}(g\tilde{g})-\frac{\partial^{2}}{\partial u\partial r}(g\tilde{g})+\tilde{g}\frac{\partial^{2}}{\partial r^{2}}(g\tilde{g})-2\tilde{g}\frac{\partial^{2}}{\partial r\partial u}g\right)+\frac{1}{2g^{2}}\left(2\tilde{g}\frac{\partial g}{\partial r}\frac{\partial g}{\partial u}-\tilde{g}\frac{\partial g}{\partial r}\frac{\partial}{\partial r}(g\tilde{g})\right) \\
       &+\frac{1}{rg}\left(\frac{\partial}{\partial u}(g\tilde{g})+\tilde{g}\frac{\partial}{\partial r}(g\tilde{g})-2\tilde{g}\frac{\partial g}{\partial u}\right)\;, \\
  R_{ur}&=\partial_ {r}\Gamma^{r}_{ur}-\partial_{u}\Gamma^{r}_{rr}+\Gamma^{r}_{ur}\left(\Gamma^{\theta}_{\theta r}+\Gamma^{\varphi}_{\varphi r}\right) \\
       &=\frac{1}{2g}\left(\frac{\partial^{2}}{\partial r^{2}}(g\tilde{g})-2\frac{\partial^{2}}{\partial u\partial r}g\right)-\frac{1}{2g^{2}}\left(\frac{\partial g}{\partial r}\frac{\partial}{\partial r}(g\tilde{g})-2\frac{\partial g}{\partial u}\frac{\partial g}{\partial r}\right)+\frac{1}{rg}\frac{\partial}{\partial r}(g\tilde{g})\;, \\
 R_{rr}&=-\partial_{r}\left(\Gamma^{\theta}_{\theta r}+\Gamma^{\varphi}_{\varphi r}\right)+\Gamma^{r}_{rr}\left(\Gamma^{\theta}_{\theta r}+\Gamma^{\varphi}_{\varphi r}\right)-\Gamma^{\theta}_{\theta r}\Gamma^{\theta}_{\theta r}-\Gamma^{\varphi}_{\varphi r}\Gamma^{\varphi}_{\varphi r}\\
        &=\frac{2}{r}\frac{1}{g}\frac{\partial g}{\partial r}\;, \\
 R_{\theta\theta}&=\partial_{u}\Gamma^{u}_{\theta\theta}+\partial_{r}\Gamma^{r}_{\theta\theta}-\partial_{\theta}\Gamma^{\varphi}_{\varphi\theta}+\Gamma^{u}_{\theta\theta}\left(\Gamma^{r}_{ru}+\Gamma^{u}_{uu}\right)+\Gamma^{r}_{\theta\theta}\Gamma^{r}_{rr}-\Gamma^{\varphi}_{\varphi\theta}\Gamma^{\varphi}_{\varphi\theta} \\
		 &=-\frac{1}{g}\frac{\partial}{\partial r}(r\tilde{g})+1\;, \\
 R_{\varphi\varphi}&=\partial_{u}\Gamma^{u}_{\varphi\varphi}+\partial_{r}\Gamma^{r}_{\varphi\varphi}+\partial_{\theta}\Gamma^{\theta}_{\varphi\varphi}+\Gamma^{u}_{\varphi\varphi}\left(\Gamma^{u}_{uu}+\Gamma^{r}_{ur}\right)+\Gamma^{r}_{\varphi\varphi}\Gamma^{r}_{rr}-\Gamma^{\theta}_{\varphi\varphi}\Gamma^{\varphi}_{\theta\varphi} \\
                &=\sin^2\theta R_{\theta\theta} \;.
\end{aligned}
\end{equation*}

\end{itemize}

The Einstein field equations \eqref{fieldEq} then have the following nontrivial components:

\begin{equation}
\begin{aligned}
\frac{1}{2g}\left(\tilde{g}\frac{\partial^{2}}{\partial r^{2}}(g\tilde{g})-2\tilde{g}\frac{\partial^{2}}{\partial r\partial u}g\right)&+\frac{1}{2g^{2}}\left(2\tilde{g}\frac{\partial g}{\partial r}\frac{\partial g}{\partial u}-\tilde{g}\frac{\partial g}{\partial r}\frac{\partial}{\partial r}(g\tilde{g})\right) \\
&+\frac{g}{r}\frac{\partial}{\partial u}\left(\frac{\tilde{g}}{g}\right)+\frac{\tilde{g}}{rg}\frac{\partial}{\partial r}(g\tilde{g})=\kappa\left(\partial_{u}\phi\right)^{2}-\Lambda g\tilde{g}\;,
\label{Einstein_uu}
\end{aligned}
\end{equation}

\begin{equation}
\begin{aligned}
\frac{1}{2g}\left(\frac{\partial^{2}}{\partial r^{2}}(g\tilde{g})-2\frac{\partial^{2}}{\partial u\partial r}g\right)&-\frac{1}{2g^{2}}\left(\frac{\partial g}{\partial r}\frac{\partial}{\partial r}(g\tilde{g})-2\frac{\partial g}{\partial u}\frac{\partial g}{\partial r}\right) \\
&+\frac{1}{rg}\frac{\partial}{\partial r}(g\tilde{g})=\kappa\left(\partial_{u}\phi\right)\left(\partial_{r}\phi\right)-\Lambda g\;,
\end{aligned}
\end{equation}

\begin{equation}
\qquad\frac{2}{r}\frac{1}{g}\frac{\partial g}{\partial r}=\kappa\left(\partial_{r}\phi\right)^{2}\;,
\end{equation}

\begin{equation}
\frac{\partial}{\partial r}(r\tilde{g})=g\left(1-\Lambda r^{2}\right)\;.
\end{equation}

The wave equation \eqref{wave} reads
\begin{equation}
\begin{aligned}
& -\frac{2}{g}\left(\partial_{u}-\Gamma^{r}_{ru}\right)(\partial_{r}\phi)+\frac{\tilde{g}}{g}\left(\partial_{r}-\Gamma^{r}_{rr}\right)(\partial_{r}\phi)-\frac{2}{r^{2}}\Gamma^{r}_{\theta\theta}(\partial_{r}\phi)-\frac{2}{r^{2}}\Gamma^{u}_{\theta\theta}(\partial_{u}\phi)=0 \\
& \Leftrightarrow\quad\frac{1}{r}\left[\frac{\partial}{\partial u}-\frac{\tilde{g}}{2}\frac{\partial}{\partial r}\right]\frac{\partial}{\partial r}\left(r\phi\right)=\frac{1}{2}\left(\frac{\partial\tilde{g}}{\partial r}\right)\left(\frac{\partial\phi}{\partial r}\right)\;.
\end{aligned}
\end{equation}

\providecommand{\bysame}{\leavevmode\hbox to3em{\hrulefill}\thinspace}
\providecommand{\MR}{\relax\ifhmode\unskip\space\fi MR }
\providecommand{\MRhref}[2]{%
  \href{http://www.ams.org/mathscinet-getitem?mr=#1}{#2}
}
\providecommand{\href}[2]{#2}

\end{document}